\documentclass[11pt,a4paper,reqno]{article}
\usepackage{amsmath}
\usepackage{graphicx}
\usepackage{amsfonts}
\usepackage{amssymb}
\usepackage{amsthm}
\usepackage[left=2.5cm, right=2.5cm, top=3cm, bottom=3.5cm]{geometry}
\usepackage{indentfirst}
\usepackage[all]{xy}
\usepackage[colorlinks=true,linkcolor=blue]{hyperref}
\usepackage{mathrsfs} %%pour le style calligraphie en mathscr
\usepackage{tikz-cd}
\usepackage{enumitem}
\usepackage[title]{appendix}
\setitemize{noitemsep,topsep=0pt,parsep=0pt,
partopsep=0pt,itemindent=12pt,leftmargin=0pt}

\newtheorem{theorem}{Theorem}[section]
\newtheorem{lemma}[theorem]{Lemma}
\newtheorem{proposition}[theorem]{Proposition}
\newtheorem{corollary}[theorem]{Corollary}

\newtheorem{conjecture}[theorem]{Conjecture}

\newcommand{\build}[3]{\mathrel{\mathop{\kern 0pt#1}\limits_{#2}^{#3}}}

\def\SU{{\mathrm{SU}}}

\def\U{{\mathrm U}}
\def\Z{{\mathbb Z}}

\def\tr{\mathrm{tr}}
\def\R{\mathbb{R}}
\def\C{\mathbb{C}}
\def\Tr{\mathrm{Tr}}

\def\vol{\mathrm{vol}}
\def\E{{\mathbb{E}}}
\def\Hom{{\mathrm{Hom}}}

\def\geq{\geqslant}
\def\leq{\leqslant}

\providecommand{\keywords}[1]
{
  \small
  \textbf{\textit{Keywords---}} #1
  \normalsize
}

\title{Almost flat highest weights and application to Wilson loops on compact surfaces}
\author{Thibaut Lemoine\thanks{Coll\`ege de France. E-mail: \texttt{thibaut.lemoine@college-de-france.fr}}}

\begin{document}
\maketitle

\begin{abstract}
We derive new formulas for the expectation and variance of Wilson loops for any contractible simple loop on a compact orientable surface of genus $1$ and higher, in the model of two-dimensional Yang--Mills theory with structure group $\U(N)$. They are written in terms of a Gaussian measure on the dual of $\U(N)$ introduced recently by the author and M. Ma\"ida \cite{LM3}. From these formulas, we prove a quantitative result on the convergence of the expectation and variance as $N$ tends to infinity, refining a result of \cite{DL}. We finally derive the large $g$ limit of the Wilson loop expectation and variance, by analogy with the study of integrals on moduli spaces of compact hyperbolic surfaces. Surprisingly, the variance does not vanish in this regime, but there are no nontrivial fluctuations of any order.
\end{abstract}


\section{Introduction}

\subsection{Yang--Mills measure and Wilson loops}

The two-dimensional quantum Yang--Mills theory is a toy model for the quantum gauge theory in four dimensions used notably in the Standard model of particle physics. It was started by Migdal \cite{Mig}, then developed mainly by Driver \cite{Dri}, Witten \cite{Wit}, Xu \cite{Xu}, Sengupta \cite{Sen4,Sen3} and L\'evy \cite{Lev3}. It can be described as a measure on the connections on a $G$-principal bundle modulo gauge transformations, and it can be reduced to a random matrix model thanks to the holonomy map. Let us briefly describe this construction, borrowed from the approach of L\'evy \cite{Lev2}. Consider a connected closed surface $\Sigma$ endowed with an area measure $\vol$, a compact group $G$ such that its Lie algebra $\mathfrak{g}$ is endowed with a $G$-invariant inner product, and an oriented topological map on $\Sigma$, that is, an oriented graph $\mathbb{G}=(\mathbb{V},\mathbb{E},\mathbb{E}^+,\mathbb{F})$ such that all faces of $\mathbb{G}$ are homeomorphic to disks. The \emph{Yang--Mills holonomy field} is a $G$-valued stochastic process $(H_\ell)_{\ell\in\mathscr{P}(\mathbb{G})}$ indexed by the set of paths in $\mathbb{G}$ obtained by concatenation of edges and their inverses. The distribution of this process can be described in terms of random variables taking values in the configuration space $\mathscr{C}_\mathbb{G}^G=G^{\mathbb{E}^+}.$ The \emph{holonomy function} is defined for a loop $\ell=e_1^{\varepsilon_1}\cdots e_n^{\varepsilon_n}$ as
\[
h_\ell:\left\lbrace \begin{array}{lll}
\mathscr{C}_\mathbb{G}^G & \to & G\\
\mathbf{g} & \mapsto & \mathbf{g}_{e_n}^{\varepsilon_n}\cdots \mathbf{g}_{e_1}^{\varepsilon_1},
\end{array}\right.
\]
and the configuration space is endowed with the smallest sigma-algebra that makes $h_\ell$ measurable for any $\ell\in\mathscr{P}(\mathbb{G})$. We also denote by $(p_t)_{t\geq 0}$ the heat kernel on $G$ for the chosen metric on $\mathfrak{g}$. For any loop $\ell$, the distribution of the random variables $H_\ell$ is given by the \emph{Driver--Sengupta formula}
\begin{equation}\label{eq:DS000}
\mathbb{E}[f(H_\ell)] = \frac{1}{Z_G(g,T)}\int_{G^{\mathbb{E}^+}} f(h_\ell(\mathbf{g})) \prod_{F\in\mathbb{F}} p_{\vol(F)}(h_{\partial F}(\mathbf{g}))d\mathbf{g},
\end{equation}
named after the two mathematicians who derived it initially: Driver \cite{Dri} and Sengupta \cite{Sen4,Sen3}. The quantity $Z_G(g,T)$ is a normalisation constant, called \emph{partition function}, which only depends on the structure group $G$, the genus $g$ of the surface and its total area $T$. In this paper, we denote by $Z_N(g,T)$ the partition function $Z_{\U(N)}(g,T)$. Using a sort of Kolmogorov extension theorem, it can be proved that this process is the finite-dimensional marginal of a more general process $(H_\ell)_{\ell\in\mathscr{P}(\Sigma)}$ indexed by a given set of paths $\mathscr{P}(\Sigma)$ on the underlying surface. However, when considering functionals of the holonomy along a simple loop, the discrete setting is sufficient, by embedding the path in a topological map.\\

The holonomy function yields another interpretation of the Yang--Mills measure, related to the work of Atiyah--Bott \cite{AB83} and Goldman \cite{Gol}, as described by Witten \cite{Wit}: if $\Gamma=\Gamma(\mathbb{G})$ is the group of reduced loops on $\mathbb{G}$ with base $v\in V$ \cite{Lev2}, the discrete Yang--Mills measure on $\mathbb{G}$ is a non-uniform measure on the moduli space
\[
\Hom(\Gamma,G)/G.
\]
For instance, if $\mathbb{G}$ is a map with one face embedded in a compact surface $\Sigma$ of genus $g$, then $\Gamma(\mathbb{G})=\pi_1(\Sigma,v)$. This point of view draws a parallel with the integration theory on moduli spaces of Riemann surfaces, like the model of Weil--Petersson developed by Mirzakhani for hyperbolic surfaces \cite{Mir}, or the model of random representations of surface groups developed by Magee and Puder \cite{Mag,MP23}. In both models, the most relevant results are obtained in an asymptotic regime: when the genus is large in the Weil--Petersson case, and when the group is large in the case of random representations. See for instance \cite{Wri} for an overview of Mirzakhani's work on the Weil--Petersson integration, and \cite{LW2} for an overview of recent results about the large genus regime.\\

The main observables of Yang--Mills theory are complex random variables given by $\tr(H_\ell)$ for arbitrary loops $\ell$, and are called \emph{Wilson loops}. In the physics literature, Wilson loop expectations are mainly studied, but in this paper we will also consider the variance of Wilson loops. Given a simple loop $\ell$ on a closed surface $\Sigma$ of genus $g\geq 1$, \emph{i.e.} a loop without self-intersection, there are two possibilities: either $\Sigma\setminus\ell$ is connected and $\ell$ is said to be \emph{nonseparating}, or $\Sigma\setminus\ell$ contains at least two connected components and $\ell$ is said to be \emph{separating}. It is known (cf. \cite[§6.3]{Sti}) that any nonseparating loop is canonically homeomorphic to an edge of the fundamental domain\footnote{It means in particular that such a loop can be completed into a set of generators of $\pi_1(\Sigma)$.} of $\Sigma$ and that any separating loop splits $\Sigma$ into two components, which are homeomorphic to compact connected orientable surfaces with boundary, and these surfaces have respectively genus $g_1$ and $g_2$ such that $g_1+g_2=g$. Furthermore, $\ell$ is contractible if and only if it is separating and $g_1=0$ or $g_2=0$. In the present paper, we shall focus on the case of contractible simple loops: we will provide exact and asymptotic formulas for the first moments of their Wilson loops in terms of particular irreducible representations of $G$.\\

We will consider $G$ to be the unitary group $\U(N)$ for two main reasons: it is expected to behave nicely when $N$ tends to infinity, since the seminal work of \ 't Hooft \cite{Hoo}, and it permits to describe the distribution of the Yang--Mills holonomy field in terms of Young diagrams and Schur functions, which have particularly good combinatorial properties. The case of $\SU(N)$ can be done with the same tools as $\U(N)$ and we only omit its treatment to keep the paper of reasonable length; the formulas would be pretty similar and the arguments unchanged.

\subsection{Almost flat highest weights}

The irreducible representations of $\U(N)$ are labeled by nonincreasing $N$-tuples of integers $\lambda=(\lambda_1\geq\cdots\geq \lambda_N)$ called \emph{highest weights}. We denote by $\widehat{\U}(N)$ the set of highest weights. We also define:
\begin{enumerate}
\item The \emph{character} of the representation
\[
\chi_\lambda(U)=\chi_\lambda(x_1,\ldots,x_N),\ \forall U\sim\mathrm{diag}(x_1,\ldots,x_N)\in\U(N),
\]
\item The \emph{dimension} of the representation, which is $d_\lambda = \chi_\lambda(1,\ldots,1)$,
\item The \emph{Casimir number} of the representation, which is the nonnegative number $c_2(\lambda)$ such that
\[\Delta \chi_\lambda = -c_2(\lambda)\chi_\lambda.\]
\end{enumerate}
The \emph{character expansion} of the heat kernel is given by\footnote{See for instance \cite[Thm. 4.2]{Lia}.}
\begin{equation}\label{eq:HKdecomp}
p_t(U)=\sum_{\lambda\in\widehat{\U}(N)} e^{-c_2(\lambda)\frac{t}{2}}d_\lambda \chi_\lambda(U),~\forall T> 0,~\forall U\in\U(N).
\end{equation}

Several asymptotic estimations of the heat kernel on the unitary group have already been obtained \cite{Bia,Lev5,LM2}. Most of them are sufficient to study Wilson loops on the plane \cite{AS,Lev}, but not on compact surfaces. For instance, the study of the partition function on a sphere was associated to the study of a Brownian bridge on $\U(N)$ \cite{LM,DN}, and the one on higher genus compact surfaces was related to particular highest weights called \emph{almost flat} \cite{Lem}. Such highest weights are in fact couplings of Young diagrams of finite size. Their typical shape exhibits a large plateau, hence the name `almost flat' -- see Fig. \ref{fig:lambdaabc}.
\begin{figure}[h!]
    \centering
    \includegraphics[width=0.6\linewidth]{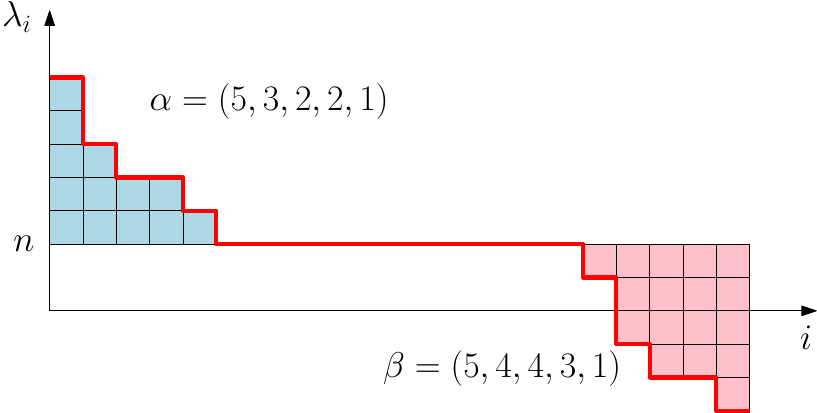}
    \caption{\small Almost flat highest weight of $\U(N)$.}
    \label{fig:lambdaabc}
\end{figure}
They can be somehow related to the so-called stable representations of the unitary group . In this case, if we denote $\lambda_N(\alpha,\beta,n)$ such highest weight, one can prove that
\[
c_2(\lambda(\alpha,\beta,n))= \vert\alpha\vert + \vert\beta\vert+n^2 + O(N^{-1}),
\]
where we noted $\vert\alpha\vert = \sum_i \alpha_i$. This formula can be interpreted as a sort of decoupling of the partitions $\alpha$, $\beta$ and $(n,\ldots,n)$, and plays a crucial role in the study of the Yang--Mills partition function on compact surfaces. This coupling/decoupling was recently further analyzed by the author and M. Ma\"ida  to get an asymptotic expansion of the partition function for $g=1$. To this end, they introduced a Gaussian measure on $\widehat{\U}(N)$ defined by
\[
\mathbb{P}_{\widehat{\U}(N),T}(\{\lambda\}) = \frac{1}{Z_N(1,T)}e^{-\frac{T}{2}c_2(\lambda)}, \quad \forall \lambda\in\widehat{\U}(N).
\]
This measure will prove useful as well for the study of Wilson loops. We will denote by $\E_{\widehat{\U}(N),T}$ the associated expectation.

\subsection{Main results}

If $\ell$ is a contractible simple loop on a surface $\Sigma_{g,T}$, then it is the boundary of a topological disk $D$ of area $t\in(0,T)$, \emph{i.e.} a two-dimensional topological manifold with boundary homeomorphic to a closed disk; $t$ will be called the \emph{interior area} of $\ell$. If we remark that $\Sigma_{g,T}\setminus D$ is homeomorphic to a surface $\Sigma'$ with boundary, then we can choose a set of generators $\{a_1,b_1,\ldots,a_g,b_g\}$ of $\pi_1(\Sigma')$, and we can pull them back by homeomorphism into generators of $\pi_1(\Sigma_{g,T})$. By taking the base point $v_1$ of these generators and the base point $v_2$ of $\ell$, and considering a simple curve $e$ from $v_1$ to $v_2$, we have that $\{a_1,b_1,\ldots,a_g,b_g,e,\ell\}$ is the set of edges of a graph $\mathbb{G}$ with two faces of respective areas $t$ and $T-t$. Such a graph is illustrated in Fig. \ref{fig:lacet-simple-disque} for a surface $\Sigma_{2,T}$ of genus $2$.\\

\begin{figure}[!h]
\centering
\includegraphics[scale=0.8]{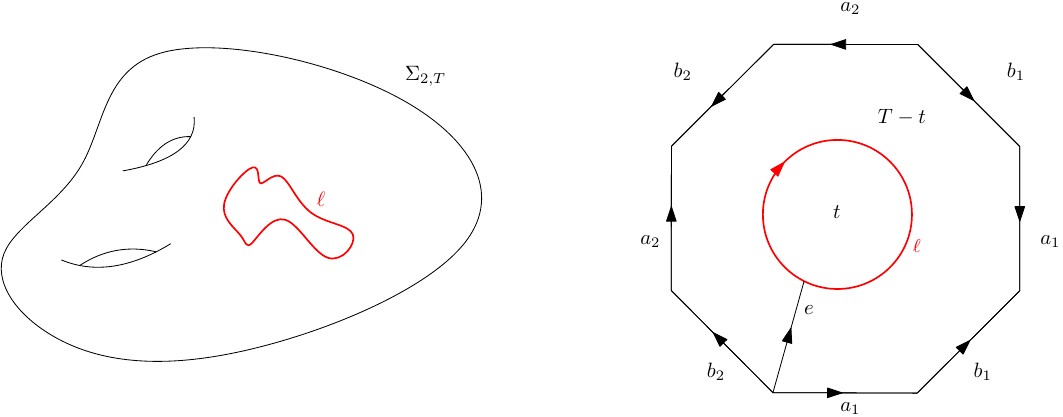}
\caption{\small An contractible simple loop $\ell$ (on the left) and the oriented graph associated to it (on the right).}\label{fig:lacet-simple-disque}
\end{figure}

If $\ell$ is a simple loop of interior area $t$, then we can compute its Wilson loop expectation using Driver--Sengupta formula \eqref{eq:DS000}:

\begin{equation}\label{eq:loop1}
\E[\tr(H_\ell)]=\frac{1}{Z_N(g,T)}\int_{G^{2g+1}} \tr(x) p_t(x^{-1}) p_{T-t}(x[y_1,z_1]\cdots[y_g,z_g])dx\prod_{i=1}^gdy_idz_i,
\end{equation}
where $(p_t)_{t\in\R_+}$ is the heat kernel on $G$.\\

We can also define the \emph{Wilson loop variance} as the variance of Wilson loop functional -- as tautological as it seems:
\[
\mathrm{Var}[\tr(H_\ell)] = \mathbb{E}[|\tr(H_\ell)|^2]-|\mathbb{E}[\tr(H_\ell)]|^2.
\]
Hence, the variance can be explicitly computed as long as we know the expectations $\mathbb{E}[\tr(H_\ell)]$ and $\mathbb{E}[\vert\tr(H_\ell)\vert^2]$. These two quantities will therefore be the main protagonists of this paper. We will also need an additional notation: if $\lambda$ and $\mu$ are two highest weights, we will write $\mu\nearrow\lambda$ (or $\lambda\searrow\mu$) if $\lambda$ can be obtained from $\mu$ by adding $1$ to one of its parts\footnote{In the language of diagrams, it means that we add a box to a positive part or remove one to a negative part.}, and $\mu\sim\lambda$ if $\lambda$ can be obtained from $\mu$ by adding $1$ to one of its parts and $-1$ to one of its parts\footnote{it can be the same one!}, \emph{i.e.} if there exists a highest weight $\nu$ such that $\lambda\nearrow\nu$ and $\mu\nearrow\nu$. These branching rules are illustrated in Fig. \ref{fig:diagramme-signature3}. The same notations are also used for branching rules over integer partitions.

\begin{figure}[h!]
\centering
\includegraphics[scale=0.7]{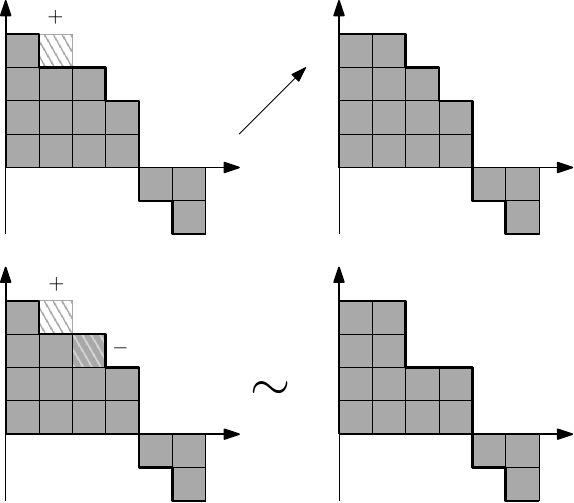}
\caption{\small On the first row: we have $\lambda\nearrow\mu$, with $\lambda$ on the left and $\mu$ on the right. On the second row: we have $\lambda\sim\mu$, with $\lambda$ on the left and $\mu$ on the right.}\label{fig:diagramme-signature3}
\end{figure}

Finally, let us introduce quantities that will follow us throughout this paper: given $0<t<T$, we set $q_t=e^{-\frac{t}{2}}$. Some formulas will be more convenient to write in terms of $q_t$ and $q_T$ instead of $t$ and $T$; this is due to a difference of terminology between Gaussian measures, which are written in terms of exponential terms, and generating series of theta and phi functions, which are rather written in terms of a complex variable $q\in\C$ such that $\vert q\vert<1$.

\begin{theorem}\label{prop:wilson_loops_exp_var}
Let $\Sigma_{g,T}$ be an orientable compact connected surface of genus $g\geq 1$ and of area $T$, $\ell$ be a contractible simple loop of interior area $t$ oriented clockwise. We have the following formulas:
\begin{align}
\mathbb{E}[\tr(H_\ell)]=& \frac{Z_N(1,T)}{Z_N(g,T)}\E_{\widehat{\U}(N),T}\left[\sum_{\substack{\mu\in\widehat{\U}(N): \mu\searrow\lambda}}\frac{d_\mu}{Nd_\lambda^{2g-1}}q_t^{c_2(\mu)-c_2(\lambda)} \right],\label{eq:wilson_loop_exp_sun}\\
\mathbb{E}[\vert\tr(H_\ell)\vert^2]=& \frac{Z_N(1,T)}{Z_N(g,T)}\E_{\widehat{\U}(N),T}\left[\sum_{\substack{\mu\in\widehat{\U}(N): \mu\sim\lambda}}\frac{d_\mu}{N^2d_\lambda^{2g-1}}q_t^{c_2(\mu)-c_2(\lambda)} \right].\label{eq:wilson_loop_var_sun}
\end{align}
\end{theorem}

This is a fairly general result; let us specialize it to $\U(1)$. In this case, the branching rules are quite degenerate: the only $\mu\searrow (n)$ for $n\in\Z\simeq\widehat{\U}(1)$ is $(n+1)$, and the only $\mu\sim(n)$ is $(n)$ itself. Furthermore, all irreducible representations have dimension 1. Hence, $Z_1(1,T)=Z_1(g,T)$ for all $g\geq 1$, and
\begin{equation}\label{eq:u1}
\E[\tr(H_\ell)]=\E_{\widehat{\U}(1),T}[e^{-nt-\frac{t}{2}}],\quad \E[\vert\tr(H_\ell)\vert^2] = 1.   
\end{equation}
It appears therefore that in the abelian case, the topology of the underlying surface has no effect on the Wilson loop.\\

From Theorem \ref{prop:wilson_loops_exp_var}, we are able to derive asymptotics in large $N$ and large $g$ limits. In the case of large $N$, we first obtain the following.

\begin{theorem}\label{thm:exp}
Let $\Sigma_{g,T}$ be an orientable compact connected surface of genus $g\geq 1$ and of area $T>0$, and $\ell$ be a contractible simple loop of interior area $t<T$. The associated Wilson loop expectation converges, as $N\to\infty$, and its limit is
\begin{equation}\label{eq:limexp}
\lim_{N\to\infty}\mathbb{E}[\tr(H_\ell)]=e^{-\frac{t}{2}}.
\end{equation}
Furthermore, we have the following fluctuations:
\begin{itemize}
\item if $g\geq 2$,
\begin{equation}\label{eq:fluct_exp}
\E[\tr(H_\ell)] = e^{-\frac{t}{2}} + O(N^{-2}).
\end{equation}
\item If $g=1$, then for any $\varepsilon>0$ small enough,
\begin{equation}\label{eq:fluct_exp2}
\E[\tr(H_\ell)] = e^{-\frac{t}{2}} + O(N^{-1+\varepsilon}).
\end{equation}
\end{itemize}
\end{theorem}

\begin{theorem}\label{thm:var}
Let $\Sigma_{g,T}$ be an orientable compact connected surface of genus $g\geq 1$ and of area $T>0$, and $\ell$ be a contractible simple loop of interior area $t<T$. The associated Wilson loop variance satisfies the following limit as $N\to\infty$:
\begin{equation}\label{eq:limvar}
\lim_{N\to\infty}\mathrm{Var}[\tr(H_\ell)]=0.
\end{equation}
Furthermore, we have the following fluctuations:
\begin{itemize}
\item if $g\geq 2$,
\begin{equation}\label{eq:fluct_var}
\mathrm{Var}[\tr(H_\ell)] = O(N^{-2}).
\end{equation}
\item If $g=1$, then for any $\varepsilon>0$ small enough,
\begin{equation}\label{eq:fluct_var2}
\mathrm{Var}[\tr(H_\ell)] = O(N^{-1+\varepsilon}).
\end{equation}
\end{itemize}
\end{theorem}

Theorems \ref{thm:exp} and \ref{thm:var} imply that the Wilson loops $W_\ell$ converge in probability to the limit given by \eqref{eq:limexp}. Note that the limit is the same as if the loop was embedded in the plane, which was already observed in \cite{DL}. A full asymptotic expansion of the Wilson loop expectation should be within reach, provided that an expansion of the partition function exists. We shall investigate this in the future, relying on recent and upcoming progresses on the partition function \cite{LM3,LM4}. Note that the fluctuations in the case of $g=1$ might not be sharp, and we conjecture that they should be of order at least $\frac1N$; however, the techniques developed in the present paper are not sufficient to prove it. The main limiting factor is Proposition \ref{prop:GTbis}.\\

In the large $g$ regime, we obtain analogs of Theorems \ref{thm:exp} and \ref{thm:var}, with a few differences. As we have seen in \eqref{eq:u1}, there is no dependence in $g$ when the structure group is $\U(1)$, hence we will only consider here $\U(N)$ with $N\geq 2$.

\begin{theorem}\label{thm:large_g}
Let $\Sigma_{g,T}$ be an orientable compact connected surface of genus $g\geq 1$ and of area $T>0$, and $\ell$ be a contractible simple loop of interior area $t<T$. For any fixed integers $N\geq 2$ and $k\geq 1$, as $g\to\infty,$
\begin{equation}\label{eq:exp_largeg}
\E[\tr(H_\ell)] = \frac{1}{\theta(q_T)} \sum_{n\in\Z} e^{-\frac{T}{2}n^2-\frac{tn}{N}-\frac{t}{2}} + O(g^{-k}).
\end{equation}

\begin{equation}\label{eq:var_largeg}
\E[\vert\tr(H_\ell)\vert^2] = \frac{1}{N^2}+ e^{-t} + O(g^{-k}).
\end{equation}
In particular, the variance of Wilson loops does not vanish in the large $g$ limit.
\end{theorem}

% We first provide several (mostly well-known) results about the representation theory of $\U(N)$ in Section \ref{sec:Prelim}. Then, we study in Section \ref{sec:Wilson} the Wilson loops for general values of $N$ and $g$. The main result we get is Theorem \ref{prop:wilson_loops_exp_var}, where we give an explicit expression, in terms of a discrete Gaussian measure on $\widehat{\U}(N)$, of $\E[\tr(H_\ell)]$ and $\E[\vert\tr(H_\ell)\vert^2]$ for a contractible simple loop $\ell$.

% In Section \ref{sec:LargeN} we let $N$ tend to infinity in our main formulas, which leads to Theorems \ref{thm:exp} and \ref{thm:var}. In particular, the limit of $\E[\tr(H_\ell)]$ is the same as in the plane, and we are able to estimate the fluctuations of its convergence.

% Finally, in Section \ref{sec:LargeG} we let $g\to\infty$, with fixed $N$, and obtain Theorem \ref{thm:large_g}, which is the analog of Theorem \ref{thm:exp}. In this regime, the limit of $\E[\tr(H_\ell)]$ is the same as in $\U(1)$, and we are again able to estimate its fluctuations. Strikingly, the variance of the Wilson loop does \emph{not} vanish asymptotically.

\subsection{Related works}

\subsubsection{Yang--Mills theory}

The limits obtained in Theorems \ref{thm:exp} and \ref{thm:var} were conjectured in a paper by Hall:

\begin{conjecture}[\cite{Hal2}]
If $\Sigma$ is a compact surface of area $T$ and $C$ is a simple loop in a topological disk, with interior area $t<T$, then the limit $\E[\tr(H_\ell)]$ exists for $\U(N)$ as $N\to\infty$, and it depends continuously on $t$ and $T-t$. Moreover, $\lim\mathrm{Var}(\tr(H_\ell))=0$.
\end{conjecture}

The first proof of this conjecture was given in \cite{DL}, where Dahlqvist and the author extended the result to any contractible loop (possibly with self-intersections) and any classical group. The proofs of \cite{DL} relied on absolute continuity arguments, and estimates of moments of Brownian motions on $\U(N)$. Theorem \ref{thm:exp} and \ref{thm:var} provide another proof of Hall's conjecture, which is more direct, and we improve the convergence results of \cite{DL} with the estimates \eqref{eq:fluct_exp} and \eqref{eq:fluct_exp2}. The results of \cite{DL} have been extended in \cite{DL2} to general loops in a torus, and a large class of loops in higher genus surfaces, but the case of general loops for $g\geq 2$ remains an open problem.\\

If the compact surface is replaced by the plane, the convergence and fluctuations of Wilson loops have been successfully studied through free probability and Schur--Weyl duality \cite{AS,Xu,Lev}. More recently, Park--Pfeffer--Sheffield--Yu \cite{PPSY} gave a new proof of the convergence of Wilson loop expectations by using techniques from random surfaces. Note that in the case of the plane, the convergence of Wilson loops for general loops has already been completely solved in \cite{Lev}.\\

Let us finally mention a few recent mathematical developments of Yang--Mills theory in other directions: SPDE techniques have been applied to the Langevin dynamics of Yang--Mills measure in two dimensions \cite{CCHS,BC}, as well as lattice Yang--Mills coupled with a Higgs field \cite{SZZ}, and a lot of progress has also been done in lattice Yang--Mills in higher dimensions \cite{CaoChatt,CPS,BCS}.

\subsubsection{Random representations of surface groups}

Analogous results have been obtained in the model of random representations of $\Gamma_g=\pi_1(\Sigma_g)$ for a compact surface $\Sigma_g$ of genus $g\geq 2$, with values in $G=S_N$ \cite{MP23} or $G=\SU(N)$ \cite{Mag}. Given a loop $\gamma$ defined up to homotopy, we set
\[
H_\gamma:\rho\in\Hom(\Gamma_g,G)\mapsto \rho(\gamma).
\]
This is the analog of the holonomy map in Yang--Mills theory, and the associated (unnormalized) Wilson loop
\[
\Tr(H_\gamma):\rho\mapsto \Tr(\rho(\gamma))
\]
descends to a function on the moduli space $\Hom(\Gamma_g,G)/G$. When $G=S_N$, the Wilson loops are integrated with respect to the uniform probability measure on the finite set $\Hom(\Gamma,S_N)$; when $G=\SU(N)$, they are integrated with respect to the Atiyah--Bott--Goldman measure, which is the measure inherited from the symplectic form introduced by Goldman \cite{Gol} on the moduli space of flat $\SU(N)$-connections, and which can be seen as the weak limit of the Yang--Mills measure on a compact surface of genus $g\geq 2$ and area $t$ as $t\to 0$ \cite{Wit,Sen6}.

\begin{theorem}[\cite{MP23}]
If we denote by $\E_{g,N}$ the expectation with respect to the uniform measure on $\Hom(\Gamma_g,S_N)$, for any $\gamma\in\Gamma_g$ there exist an infinite sequence of rational numbers
\[
a_1(\gamma),a_0(\gamma),a_{-1}(\gamma),\ldots
\]
such that for any $k\geq 1$, as $N\to\infty$,
\begin{equation}
\E_{g,N}[\Tr(H_\gamma)]=a_1(\gamma)N+a_0(\gamma)+\frac{a_{-1}(\gamma)}{N}+\cdots+O(N^{-k}).
\end{equation}
Furthermore, if $\gamma\in\Gamma_g$ is not the identity, then $a_1(\gamma)=0$.
\end{theorem}

\begin{theorem}[\cite{Mag}]
If we denote by $\E'_{g,N}$ the expectation with respect to the Atiyah--Bott--Goldman measure, for any $\gamma\in\Gamma_g$ there exist an infinite sequence of rational numbers
\[
a_1(\gamma),a_0(\gamma),a_{-1}(\gamma),\ldots
\]
such that for any $k\geq 1$, as $N\to\infty$,
\begin{equation}
\E'_{g,N}[\Tr(H_\gamma)]=a_1(\gamma)N+a_0(\gamma)+\frac{a_{-1}(\gamma)}{N}+\cdots+O(N^{-k}).
\end{equation}
\end{theorem}

\begin{theorem}[\cite{MagII}]
Under the assumptions of the previous theorem, if $\gamma\in\Gamma_g$ is not the identity, then $a_1(\gamma)=0$.
\end{theorem}

In a certain way, the results of Magee and Puder focus on a problem that is orthogonal to ours: all loops considered in the current paper are homotopic to the identity, which has a trivial Wilson loop in their model, whereas they rather study loops with nontrivial homotopy -- which we do not consider here. In fact, combining results of the present paper with those of \cite{Mag} are probably sufficient to estimate the speed of convergence of Wilson loops for all simple loops for the Yang--Mills measure on a compact surface of genus $g\geq 2$.

\subsubsection{Weil--Petersson model of random hyperbolic surfaces}

In the Weil--Petersson model, the representation space of a compact hyperbolic surface $\Sigma$ of genus $g$ with $n$ geodesic boundary components is the \emph{Teichm\"uller space}
\[
\mathcal{T}_{g,n} = \Hom(\pi_1(\Sigma),\mathrm{PSL}(2,\R))^{*}/\mathrm{PSL}(2,\R).
\]
It is the open subset of the full representation space consisting in discrete injective representations of $\pi_1(\Sigma)$. The geodesic lengths $L_1,\ldots,L_n$ of the boundary components define a set of parameters for the corresponding Teichm\"uller space, and we write $\mathcal{T}_{g,n}(L_1,\ldots,L_n)$ to underline the dependence on the geodesic lengths. The representation space admits a symplectic volume form $\mathrm{Vol}_{g,n}$ that defines the \emph{Weil--Petersson measure} on the Teichm\"uller space. If one quotients the Teichm\"uller space by the mapping class group $\mathrm{MCG}_{g,n}$ of the surface, one obtains the \emph{moduli space}
\[
\mathcal{M}_{g,n}(L_1,\ldots,L_n)=\mathcal{T}_{g,n}(L_1,\ldots,L_n)/\mathrm{MCG}_{g,n}.
\]
The Weil--Petersson measure has an infinite mass on the Teichm\"uller space, but it descends on the moduli space to a finite measure. The Weil--Petersson volume on $\mathcal{M}_{g,n}$ is the analog of the Yang--Mills partition function, and analogs of Wilson loops for simple loops are the so-called \emph{geometric functions}, which are functions $f_\gamma:\mathcal{M}_{g,n}\to\R$ defined by
\[
f_\gamma(X)=\sum_{[\alpha]\in\mathcal{O}(\gamma)}f(\ell_{\alpha}(X)),
\]
where $\mathcal{O}_\gamma$ is the set of homotopy classes in the orbit of $\gamma$ on $X\in\mathcal{M}_{g,n}$ by the action of the mapping class group, $\ell_\alpha(X)$ is the length of the geodesic representative of $[\alpha]$ in $Y\in\mathcal{T}_{g,n}$, and $f:\R_+\to\R$ is a given function. Their expectation with respect to the Weil--Petersson measure is by definition
\[
\E_{WP}[f_\gamma(X)]= \frac{1}{\vol_{WP}(\mathcal{M}_g)}\int_{\mathcal{M}_g}f_\gamma(X)d\vol_{WP}(X).
\]

Mirzakhani described in \cite{Mir} how to compute such expectations for geometric functions, and investigated in \cite{Mir13} the asymptotic behaviour of these integrals in the large $g$ limit. For instance, if $c$ is a separating simple loop on a compact surface $\Sigma$ of genus $g$ without boundary that splits $\Sigma$ into $\Sigma_{1,1}$ and $\Sigma_{g-1,1}$ with one boundary component and respective genera 1 and $g-1$. Let
\[
F_c^L(X)=\vert\{\gamma\in\mathcal{O}_c,\ell_\gamma(X)\leq L\}\vert
\]
be the counting function of closed geodesics in the $\mathrm{MCG}_{g,0}$-orbit of a given closed curve $c$ with length at most $L$. Then Mirzakhani proved in \cite{Mir13}
\[
\E_{WP}[F_c^L(X)] \sim \frac{e^{L/2}L^3}{g},
\]
as $g\to\infty$. Further results have been obtained since then \cite{MZ,MirPet,AM,LW2}, leading to asymptotic expansions of expectations of geometric functions in powers of $\frac1g$. The main motivation in these papers is to describe the spectral gap of random hyperbolic surfaces with respect to the Weil--Petersson measure. In comparison, Theorem \ref{thm:large_g} shows that after the leading term, there is no polynomial term in $\frac1g$ of any order, which seems to be a surprising difference. Note that, however, no contractible simple loop has any use in the Weil--Petersson model: one rather deals with nonseparating simple loops or non-contractible separating simple loops, in order to find a geodesic representative of its homotopy class. 

\subsection*{Acknowledgements}

I would like to thank Thierry L\'evy for many remarks and corrections in a preliminary version of this article, Justine Louis for several comments on the Casimir number estimates, and Thierry Gobron for suggesting the study of large genus asymptotics. I would also like to thank Nalini Anantharaman, Antoine Dahlqvist, Michael Magee, and Myl\`ene Ma\"ida for several discussions. I also thank the anonymous referee for valuable comments. Part of this work has been financially supported by ANR AI chair BACCARAT (ANR-20-CHIA-0002).

\tableofcontents

\section{Preliminary tools}\label{sec:Prelim}

In this section, we study irreducible representations of $\U(N)$ as couplings of integer partitions. The irreducible representations of $\U(N)$ are labelled by nonincreasing $N$-tuples of integers $\lambda=(\lambda_1,\ldots,\lambda_N)\in\Z^N$ with $\lambda_1\geq\cdots\geq\lambda_N$ (we will use the notation $\lambda=(\lambda_1\geq\cdots\geq\lambda_N)$ for nonincreasing $N$-tuples) called \emph{highest weights}, and we denote respectively by $d_\lambda$ and $c_2(\lambda)$ their \emph{dimension} and \emph{Casimir number}, given by
\begin{equation}\label{eq:dim_un2}
d_\lambda=\prod_{1\leq i<j\leq N} \frac{\lambda_i-\lambda_j+j-i}{j-i}
\end{equation}
and
\begin{equation}\label{eq:c2_un2}
c_2(\lambda) = \frac{1}{N}\left(\sum_{i=1}^N \lambda_i^2 + \sum_{1\leq i<j\leq N} (\lambda_i-\lambda_j)\right).
\end{equation}
The set of irreducible representations is denoted by $\widehat{\U}(N)$ and is in bijection with the set of highest weights. The \emph{character} of a representation of highest weight $\lambda$ is given by the \emph{Schur function} $s_\lambda$, which is a symmetric polynomial when $\lambda_N\geq 0$ and a symmetric Laurent polynomial otherwise. We will not need its explicit formula for our computations, but refer to \cite{Mac} and \cite{Sta} for details about it. The character decomposition \eqref{eq:HKdecomp} of the heat kernel on $\U(N)$ becomes then
\begin{equation}\label{eq:FourierUN2}
p_T(U)=\sum_{\lambda\in\widehat{\U}(N)} e^{-c_2(\lambda)\frac{T}{2}}d_\lambda s_\lambda(U),~\forall T> 0,~\forall U\in\U(N).
\end{equation}

In order to compute Wilson loop expectations and variances in a convenient way, we will need two ingredients: a specific measure on $\widehat{\U}(N)$ that was introduced in \cite{DL} and studied in detail in \cite{LM3}, and branching rules on highest weights.

\subsection{Gaussian measure and observables of highest weights}

We will consider the following measure on $\widehat{\U}(N)$:
\[
\mathbb{P}_{\widehat{\U}(N),T}(\{\lambda\}) = \frac{1}{Z_N(1,T)}e^{-\frac{T}{2}c_2(\lambda)}, \quad \forall \lambda\in\widehat{\U}(N).
\]
This measure can be considered as an analog of the discrete Gaussian measure on $\Z$, which is $\mathbb{P}_{\widehat{\U}(1),T}$, as already observed in \cite{DL,LM3}. In this particular case, we also find that $Z_{1}(1,T)=\theta(e^{-\frac{T}{2}})$, where $\theta$ is the following Jacobi theta function:
\[
\theta(q)=\sum_{n\in\Z}q^{n^2}.
\]

From two integer partitions $\alpha=(\alpha_{1}\geq \cdots \geq \alpha_{r}> 0)$ and $\beta=(\beta_{1}\geq \cdots \geq \beta_{s}> 0)$ of respective lengths $\ell(\alpha)=r$ and $\ell(\beta)=s$, and an integer $n\in \Z$, we can form, for all $N\geq r+s+1$, the \emph{composite} highest weight
\begin{equation}\label{eq:afhwbis}
\lambda_{N}(\alpha,\beta,n)=(\alpha_1+n,\ldots,\alpha_r+n,\underbrace{n,\ldots,n}_{N-r-s},n-\beta_s,\ldots,n-\beta_1) \in \widehat\U(N).
\end{equation}
This definition is extended in the obvious way to the cases where one or both of the partitions $\alpha$ and $\beta$ are the empty partition.\\

We proved in \cite{Lem} that this construction can be reversed: given a highest weight $\lambda\in\widehat{\U}(N)$, we can define unambiguously $\alpha,\beta$ and $n$ such that $\lambda=\lambda_N(\alpha,\beta,n)$. It provides a bijection $\lambda_N:\Lambda_N^0\to\widehat{\U}(N),$ where $\Lambda_N^0$ is a subset of $\mathfrak{Y}\times\mathfrak{Y}\times\Z$ given by $(\alpha,\beta,n)$ such that $\ell(\alpha)\leq\lfloor (N+1)/2\rfloor$ and $\ell(\beta)\leq N-\lfloor(N+1)/2\rfloor$. The main weights of interest will be $\lambda_N(\alpha,\beta,n)$ such that $\alpha$ and $\beta$ are small perturbations of the constant weight $(n,\ldots,n)$; such weights were already studied by Gross and Taylor in \cite{GT} for finite $\alpha$ and $\beta$. Our approach actually enables $\alpha$ and $\beta$ to grow with $N$ while remaining small enough, so that $\lambda_N(\alpha,\beta,n)$ remains close to $(n,\ldots,n)$, \emph{i.e.}, ``almost flat".\\

More recently, we also proved, together with Myl\`ene Ma\"ida, that the coupling of partitions can also pass to measures. For any $q\in(0,1)$, the $q$-uniform measure $\mathscr{U}(q)$ on the set $\mathfrak{Y}$ of integer partitions is given by
\[
\mathbb{P}(\{\alpha\}) = \phi(q)q^{\vert\alpha\vert},
\]
where $\phi(q)=\prod_{m=1}^\infty(1-q^m)$ is Euler's function. Below are the main results from \cite{LM3} that we will need in the current paper. The first one expresses expectations with respect to $\mathbb{P}_{\widehat{\U}(N),T}$ in terms of random partitions.

\begin{proposition}[\cite{LM3}]\label{prop:change_variable_partition}
For any $N\geq 1$, $q\in(0,1)$, and any measurable $f:\widehat{\U}(N)\to\R$,
\begin{equation}
\E_{\widehat{\U}(N),T}[f(\lambda)]=\frac{\theta(q_T)}{Z_N(1,T)\phi(q_T)^2}\E\left[f(\lambda_N(\alpha,\beta,n))q_T^{\frac{2}{N}F(\alpha,\beta,n)}\mathbf{1}_{\Lambda_N^0}(\alpha,\beta,n)\right],
\end{equation}
where $\alpha,\beta,n$ are independent random variables such that $\alpha,\beta\sim\mathscr{U}(q_T)$ and $n\sim\mathbb{P}_{\widehat{\U}(1),T}$.
\end{proposition}

The next two results are general estimates on partitions.

\begin{proposition}[\cite{LM3}]\label{prop:dev_ineq}
Let $q\in(0,1)$ be fixed, and $\alpha\sim\mathscr{U}(q)$ be a random partition. Then, for any $M>0,$
\begin{equation}
\mathbb{P}(\ell(\alpha) > M)\leq C_q e^{M\log q}
\end{equation}
\begin{equation}
\mathbb{P}(\vert\alpha\vert > M)\leq C_qe^{\frac{M}{2}\log q}.
\end{equation}
\end{proposition}

\begin{lemma}[\cite{LM3}]\label{lem:domin}
For any $\alpha,\beta,n$ and any $N\geq \ell(\alpha)+\ell(\beta)+1$,
\begin{equation}
c_2(\lambda_N(\alpha,\beta,n))\geq \frac12(\vert\alpha\vert+\vert\beta\vert)+ \left(n+\frac{\vert\alpha\vert-\vert\beta\vert}{N}\right)^2=:C_N(\alpha,\beta,n),
\end{equation}
and for any $q\in(0,1)$,
\begin{equation}
\sup_{N\geq 1}\sum_{\alpha,\beta,n:}q^{C_N(\alpha,\beta,n)}<\infty.
\end{equation}
\end{lemma}

We will also need a convergence estimate for $\mathbb{P}_{\widehat{\U}(1),T}$.

\begin{proposition}\label{prop:fluct_u1}
\begin{equation}\label{eq:fluct_u1}
\E_{\widehat{\U}(1),T}\left[e^{-\frac{tn}{N}}-\theta(q_T)\right] = O(N^{-2}).
\end{equation}
\end{proposition}

\begin{proof}
Using a Taylor expansion of the exponential with integral remainder, we have
\begin{align*}
\left(\sum_{n\in\Z} e^{-\frac{T}{2}n^2-\frac{t}{N}n}\right) - \theta(q_T) = & \sum_{n\in\Z} e^{-\frac{T}{2}n^2}\left(e^{-\frac{t}{N}n}-1\right)\\
= & \sum_{n\in\Z}e^{-\frac{T}{2}n^2}\left(-\frac{t}{N}n+\frac{t^2n^2}{N^2}+I_N(n)\right),
\end{align*}
where $I_N(n)$ is a quantity that is $O(N^{-3})$ uniformly in $n$. Hence,
\[
\left(\sum_{n\in\Z} e^{-\frac{T}{2}n^2-\frac{t}{N}n}\right) - \theta(q_T) = -\frac{t}{N}\sum_{n\in\Z}nq_t^{n^2} + \frac{t^2}{N^2}\sum_{n\in\Z}n^2q_t^{n^2} + O(N^{-3}).
\]
The first sum vanishes by symmetry, and the second one is bounded. Equation \eqref{eq:fluct_u1} follows by dividing both terms by $\theta(q_T)$.
\end{proof}

\subsection{Branching rules}

Let us describe how the branching rules $\mu\nearrow\lambda,\mu\searrow\lambda$ and $\mu\sim\lambda$ described in the introduction are translated into the decomposition $\lambda_N(\alpha,\beta,n)$.

\begin{proposition}\label{prop:branch_afhw}
Let $\alpha,\beta,\alpha',\beta'$ be integer partitions, and $n,n',N$ three integers such that $N\geq \max(\ell(\alpha)+\ell(\beta)+1,\ell(\alpha')+\ell(\beta')+1)$. Then the following assertions are equivalent:
\begin{enumerate}
\item $\lambda_N(\alpha',\beta',n')\searrow\lambda_N(\alpha,\beta,n)$,
\item ($\alpha'\searrow\alpha$, $\beta'=\beta$ and $n'=n$) \ or \ ($\beta'\nearrow\beta$, $\alpha'=\alpha$ and $n'=n$).
\end{enumerate}
\end{proposition}

\begin{proof}[Proof]
In order to see the equivalence, recall the construction of $\lambda_N(\alpha,\beta,n)$ given in \eqref{eq:afhwbis}:
\[
\lambda_{N}(\alpha,\beta,n)=(\alpha_1+n,\ldots,\alpha_r+n,\underbrace{n,\ldots,n}_{N-r-s},n-\beta_s,\ldots,n-\beta_1)=(\lambda_1,\ldots,\lambda_N).
\]
The only way of having $\lambda_N(\alpha',\beta',n')\searrow\lambda_N(\alpha,\beta,n)$ is to increment a coefficient $\lambda_i$ such that $\lambda_i>\lambda_{i+1}$. It clearly excludes the coefficients $\lambda_{r+1},\ldots,\lambda_{N-s}$. Two only ways remain: either we increment one of the coefficients $\lambda_1,\ldots,\lambda_r$, or we increment one of the coefficients $\lambda_{r+s+1},\ldots,\lambda_N$. The first case corresponds to $\alpha'\searrow\alpha$ and the second one to $\beta'\nearrow\beta$ (while leaving the other parameters unchanged), according to the description of the coefficients in terms of $\alpha,\beta$ and $n$. The equivalence follows immediately.
\end{proof}

Using these informations, we can estimate the difference of Casimir numbers from two branching highest weights. Recall that the \emph{content} of the box $(i,j)$ of a diagram corresponding to the partition $\alpha=(\alpha_1\geq\cdots\geq\alpha_r>0)$ is the quantity $c(i,j)=j-i$, and that the \emph{total content} $K(\alpha)$ of the diagram is the sum of contents of all its boxes. We denote also by $|\alpha|$ the sum of the parts of the partition: $|\alpha|=\alpha_1+\cdots+\alpha_r$.

\begin{proposition}[\cite{Lem}]\label{prop04}
Let $\alpha$ and $\beta$ be two partitions of respective lengths $r$ and $s$. Let $n$ be an integer. Then, provided $N\geq r+s+1$, we have
\begin{equation}\label{eq15}
c_2(\lambda_N(\alpha,\beta,n)) = \vert\alpha\vert + \vert\beta\vert + n^2 +  \frac{2}{N}\big(K(\alpha)+K(\beta)+n(\vert\alpha\vert-\vert\beta\vert)\big).
\end{equation}
\end{proposition}

Here is how Casimir numbers behave with respect to branching rules.

\begin{proposition}\label{prop:casimir2}
Let $\alpha,\beta$ be two integer partitions and $N\geq\ell(\alpha)+\ell(\beta)+1$ be an integer. Then, for any $n\in\Z$ and any $\alpha'\searrow\alpha$, if we denote by $\square'$ the box that has been added to $\alpha$, we have
\begin{equation}
c_2(\lambda_N(\alpha',\beta,n))-c_2(\lambda_N(\alpha,\beta,n)) = 1 + \frac{2}{N}(c(\square')+n),
\end{equation}
For any $\beta'\nearrow\beta$, if we denote by $\square'$ the box that has been removed from $\beta$, we have
\begin{equation}
c_2(\lambda_N(\alpha,\beta',n))-c_2(\lambda_N(\alpha,\beta,n)) = -1 + \frac{2}{N}(-c(\square')+n),
\end{equation}
For any $\alpha'\searrow\alpha$ and $\beta'\searrow\beta$, if we denote respectively $\square_1$ and $\square_2$ the box added to $\alpha$ and $\beta,$
\begin{equation}
c_2(\lambda_N(\alpha',\beta',n))-c_2(\lambda_N(\alpha,\beta,n)) = 2+\frac{2}{N}(c(\square_1)+c(\square_2)).
\end{equation}
\end{proposition}

Now, let us turn to dimensions and their ratios. We first need to introduce a particular regime to deal with dimensions. Let us fix a real $\gamma\in (0,\frac13)$, that we can consider as a control parameter\footnote{Intuitively, we want it to be as small as possible, while remaining positive.}. Set
\[
\Lambda_N^\gamma=\{(\alpha,\beta,n)\in\mathfrak{Y}\times\mathfrak{Y}\times\Z: \vert\alpha\vert\leq N^\gamma,\vert\beta\vert\leq N^\gamma\}.
\]
A crucial point is the following: for $N$ large enough, any partition of an integer not greater than $N^{\gamma}$ has less than $\frac{N}{2}$ positive parts. Thus, $\Lambda_N^\gamma\subset\Lambda_N^0$. The set $\{\lambda_N(\alpha,\beta,n):(\alpha,\beta,n)\in\Lambda_N^\gamma\}$ is a subset of $\widehat{\U}(N)$ for $N$ large enough, and correspond to the heuristic notion of `almost flat' highest weights. Sometimes, it will be sufficient to consider highest weights that are truly flat, and we denote by
\[
\Lambda_N=\{\lambda_N(\varnothing,\varnothing,n),n\in\Z\}\simeq\Z
\]
the corresponding set. The dimensions of the corresponding irreducible representations can be estimated in terms of representations of the symmetric group. Recall that the irreducible representations of $S_n$ are in one-to-one correspondence with integer partitions $\alpha\vdash n$. We denote by $d^\alpha$ (resp. $\chi^\alpha$) the dimension (resp. character) of the irreducible representation of $S_n$ corresponding to $\alpha$.

\begin{proposition}\label{prop:GTbis}
Let $\alpha$ and $\beta$ be two integer partitions, $N\geq \ell(\alpha)+\ell(\beta)+1$ an integer and $\gamma\in(0,\tfrac13)$ a real number. Let us assume that $|\alpha|\leq N^\gamma$ and $|\beta|\leq N^\gamma$. The partitions $\alpha$ and $\beta$ induce two highest weights of $\U(N)$, $\tilde{\alpha}=\lambda_N(\alpha,\varnothing,0)$ and $\tilde{\beta}=\lambda_N(\beta,\varnothing,0)$. We have the following facts.
\begin{enumerate}
\item For $N$ large enough,
\begin{equation}
\frac{d^\alpha N^{|\alpha|}}{|\alpha|!} (1-2N^{2\gamma-1}) \leq d_{\tilde{\alpha}} \leq \frac{d^\alpha N^{|\alpha|}}{|\alpha|!} (1+2N^{2\gamma-1}),
\end{equation}
and the same result holds for $\beta$.
\item For any $n\in\Z$, the dimension of $\lambda_N(\alpha,\beta,n)$ satisfies the following estimation, assuming that $N$ is large enough:
\begin{equation}\label{eq:dimfactor}
\frac{d^\alpha d^\beta N^{|\alpha|+|\beta|}}{|\alpha|!|\beta|!}(1-24N^{3\gamma-1})\leq d_{\lambda_N(\alpha,\beta,n)} \leq \frac{d^\alpha d^\beta N^{|\alpha|+|\beta|}}{|\alpha|!|\beta|!}(1+24N^{3\gamma-1})
\end{equation}
\end{enumerate}
\end{proposition}

Note that this proposition generalizes a result by Gross and Taylor: in \cite{GT}, they derived similar asymptotic expansions but in the case where $|\alpha|$ and $|\beta|$ were finite and not depending on $N$. However, we really need the stronger assumption $|\alpha|,|\beta|\leq N^\gamma$ as we will see later.

\begin{proof}[Proof of Proposition \ref{prop:GTbis}]
$(i)$ let us first recall that (cf. \cite{GT,OV})
\begin{equation}\label{eq:dimgt}
d_{\tilde{\alpha}}=\frac{d^\alpha N^{|\alpha|}}{|\alpha|!}\prod_{\substack{1\leq i\leq r\\1\leq j\leq \alpha_i}}\left(1+\frac{j-i}{N}\right).
\end{equation}
But for any $1\leq i\leq r$ and $1\leq j\leq\alpha_i$, we have $1-r\leq j-i\leq \alpha_i-1$, and under the assumption $|\alpha|\leq N^\gamma$ it implies that $|j-i|\leq N^\gamma$. Thus,
\[(1-N^{\gamma-1})^{|\alpha|}\leq \prod_{\substack{1\leq i\leq r\\1\leq j\leq \alpha_i}}\left(1+\frac{j-i}{N}\right)\leq (1+N^{\gamma-1})^{|\alpha|}.\]
From the convexity inequality of the exponential function, we have
\[
(1+N^{\gamma-1})^{|\alpha|}\leq e^{|\alpha|N^{\gamma-1}}.
\]
We can use the following reverse inequalities, that hold for any $x\in(0,\tfrac12)$:
\[
e^x\leq 1+2x,\ \log(1-x)\geq -2x.
\]
It implies that, for $N$ such that $N^{\gamma-1}<\tfrac12$ (which is true for $N$ large enough),
\[
1-2N^{2\gamma-1}\leq e^{-2|\alpha|N^{\gamma-1}}\leq (1-N^{\gamma-1})^{|\alpha|}\leq \prod_{\substack{1\leq i\leq r\\1\leq j\leq \alpha_i}}\left(1+\frac{j-i}{N}\right) \leq e^{|\alpha|N^{\gamma-1}}\leq 1+2N^{2\gamma-1}.
\]
This estimation, applied to \eqref{eq:dimgt}, gives the expected result.\\

$(ii)$ Let us first remark that, from the Weyl dimension formula, we have for any $n\in\Z$
\begin{equation}\label{eq:fact_dlambda}
d_{\lambda_N(\alpha,\beta,n)} = d_{\tilde{\alpha}} d_{\tilde{\beta}} Q(\alpha,\beta),
\end{equation}
with
\[
Q(\alpha,\beta)=\prod_{\substack{1\leq i\leq r\\1\leq j\leq s}}\frac{(N+1-i-j)(\alpha_i+\beta_j+N+1-i-j)}{(\alpha_i+N+1-i-j)(\beta_j+N+1-i-j)}.
\]
As $r\leq|\alpha|$ and $s\leq|\beta|$, the assumption $|\alpha|,|\beta|\leq N^\gamma$ implies that we also have $r,s\leq N^\gamma$. For any $1\leq i\leq r$ and $1\leq j\leq s$, we have therefore
\[
-2N^\gamma\leq 3-2N^\gamma\leq \alpha_i+\beta_j-i-j+1\leq 2N^\gamma-1\leq 2N^\gamma.
\]
It implies that
\[
\left\vert \frac{1+\alpha_i+\beta_j-i-j}{N}\right\vert \leq 2N^{\gamma-1},
\]
and we have the same bound for $\left\vert\frac{1+\alpha_i-i-j}{N}\right\vert$, $\left\vert\frac{1+\beta_j-i-j}{N}\right\vert$ and $\left\vert\frac{1-i-j}{N}\right\vert$, so that
\[
Q(\alpha,\beta)=\prod_{\substack{1\leq i\leq r\\1\leq j\leq s}} \frac{(1+A_N(i,j))(1+B_N(i,j))}{(1+C_N(i,j))(1+D_N(i,j))},
\]
with $|A_N(i,j)|,|B_N(i,j)|,|C_N(i,j)|,|D_N(i,j)|\leq 2N^{\gamma-1}$.

For any $(i,j)$ we have
\[
\frac{1}{1+C_N(i,j)}=1-\frac{C_N(i,j)}{1+C_N(i,j)}=1+C'_N(i,j),
\]
with $|C'_N(i,j)|\leq 2|C_N(i,j)|$, and the same result holds for $D_N(i,j)$. It implies that
\[
Q(\alpha,\beta)=\prod_{\substack{1\leq i\leq r\\1\leq j\leq s}}(1+A_N(i,j))(1+B_N(i,j))(1+C'_N(i,j))(1+D'_N(i,j)),
\]
with $|A_N(i,j)|,|B_N(i,j)|,|C'_N(i,j)|,|D'_N(i,j)|\leq 4N^{\gamma-1}$. Hence, using the same inequalities as in $(i)$ we get the estimation
\[
e^{-8N^{3\gamma-1}}\leq (1-4N^{\gamma-1})^{N^{2\gamma}}\leq Q(\alpha,\beta) \leq (1+4N^{\gamma-1})^{N^{2\gamma}}\leq e^{4N^{3\gamma-1}},
\]
which implies
\[
1-8N^{3\gamma-1} \leq Q(\alpha,\beta) \leq 1+8N^{3\gamma-1}.
\]
We can apply this, as well as the point $(i)$, to get for $N^{2\gamma-1}<\tfrac14$ (which is in particular true for $N$ large enough)
\[
d_{\lambda_N(\alpha,\beta,n)}\leq \frac{d^\alpha d^\beta N^{|\alpha|+|\beta|}}{|\alpha|!|\beta|!}(1+2N^{2\gamma-1})^2(1+8N^{3\gamma-1}),
\]
which can be simplified considering that for $(x,y,z)\in(0,\tfrac14)^3$ such that $x+y+z<\tfrac14$,
\[
(1+x)(1+y)(1+z)\leq e^{x+y+z}\leq 1+2(x+y+z).
\]
Indeed, for $N$ such that $N^{3\gamma-1}+2N^{2\gamma-1}<\tfrac14$ we obtain
\[
d_{\lambda_N(\alpha,\beta,n)}\leq \frac{d^\alpha d^\beta N^{|\alpha|+|\beta|}}{|\alpha|!|\beta|!}(1+8N^{2\gamma-1}+16N^{3\gamma-1})\leq \frac{d^\alpha d^\beta N^{|\alpha|+|\beta|}}{|\alpha|!|\beta|!}(1+24N^{3\gamma-1}),
\]
and
\[
d_{\lambda_N(\alpha,\beta,n)}\geq \frac{d^\alpha d^\beta N^{|\alpha|+|\beta|}}{|\alpha|!|\beta|!}(1-24N^{3\gamma-1}),
\]
which proves the result.
\end{proof}

This general estimate can be used for a more specific one, which is actually the one we need.

\begin{lemma}\label{lem:ratio_dim}
Let $\alpha,\beta$ be two integer partitions, $N\geq\ell(\alpha)+\ell(\beta)+1$ be an integer and $\gamma\in(0,\frac13)$ be a real number. Let us assume that $\vert\alpha\vert+1\leq N^\gamma$ and $\vert\beta\vert\leq N^\gamma$. Then, for any $\alpha'\searrow\alpha$, any $\beta'\nearrow\beta$, and any $n\in\Z$,
\begin{equation}
\frac{d_{\lambda_N(\alpha',\beta,n)}}{d_{\lambda_N(\alpha,\beta,n)}} = \frac{Nd^{\alpha'}}{\vert\alpha'\vert d^\alpha}(1+O(N^{3\gamma-1})),
\end{equation}
\begin{equation}
\frac{d_{\lambda_N(\alpha,\beta',n)}}{d_{\lambda_N(\alpha,\beta,n)}} = \frac{\vert\beta\vert}{N}\frac{d^{\beta'}}{d^\beta}(1+O(N^{3\gamma-1})),
\end{equation}
where the $O$ is uniform on $N,\alpha,\beta,n$.
\end{lemma}

\begin{proof}
Using Proposition \ref{prop:GTbis}, we have
\[
\frac{d_{\lambda_N(\alpha',\beta,n)}}{d_{\lambda_N(\alpha,\beta,n)}} = \frac{d^{\alpha'}N^{\vert\alpha'\vert}\vert\alpha\vert!}{d^{\alpha}N^{\vert\alpha\vert}\vert\alpha'\vert!}(1+O(N^{3\gamma-1}))=\frac{Nd^{\alpha'}}{d^\alpha\vert\alpha'\vert}(1+O(N^{3\gamma-1})),
\]
where the $O$ is indeed independent from $N$ and also from $\alpha,\beta$.
The second equality follows from a similar computation.
\end{proof}

Finally, let us mention a branching formula for dimensions of irreducible representations of the symmetric group. It will be mainly useful for the large $N$ asymptotics.

\begin{proposition}\label{prop:branch}
Let $\lambda\vdash n$ for any positive integer $n$. We have
\begin{equation}\label{eq:branchplus}
\sum_{\substack{\mu\vdash (n+1)\\\mu\searrow\lambda}} \frac{d^\mu}{(n+1)d^\lambda}=1,
\end{equation}
and
\begin{equation}\label{eq:branchminus}
\sum_{\substack{\mu\vdash (n-1)\\\mu\nearrow\lambda}} \frac{d^\mu}{d^\lambda}=1.
\end{equation}
\end{proposition}
\begin{proof}[Proof]
Let us recall the so-called branching rules on $S_n$, cf. \cite{Sag} for example:
\[
\chi^\lambda\uparrow^{S_{n+1}} = \sum_{\substack{\mu\vdash (n+1)\\\mu\searrow\lambda}} \chi^\mu,
\]
and
\[
\chi^\lambda\downarrow_{S_{n-1}} = \sum_{\substack{\mu\vdash (n-1)\\\mu\nearrow\lambda}} \chi^\mu.
\]
As the character of a restricted representation is equal to the restriction of the character, the second branching rule directly implies \eqref{eq:branchminus}. For the character of an induced representation we have the following result \cite[Eq.(3.18)]{FH}: if $G$ is a finite group and $H$ a subgroup of $G$, then for any character $\chi$ of a representation of $H$ we have
\[
\chi\uparrow^G(g)=\frac{1}{|H|}\sum_{\substack{x\in G\\xgx^{-1}\in H}} \chi(xgx^{-1}),\ \forall g\in G.
\]
If we apply this formula to $G=S_{n+1}$, $H=S_n$, $\chi=\chi^\lambda$ and $g=1$ we get
\eqref{eq:branchplus} as expected.
\end{proof}

\section{Proof of Theorem \ref{prop:wilson_loops_exp_var}}\label{sec:Wilson}

In this section we prove Theorem \ref{prop:wilson_loops_exp_var}. We first need a formula for the integration of characters over commutators of independent uniform random variables.

\begin{lemma}\label{lem:intcommu}
Let $G$ be a compact group and $dg$ its normalized Haar measure. If $(\rho,V)$ is an irreducible representation of $G$, we have
\begin{equation}\label{eq02}
\int_{G^2} \chi_\rho(x[y,z])dydz = \frac{\chi_\rho(x)}{d_\rho^2},\ \forall x\in G.
\end{equation}
\end{lemma}

In order to prove Lemma \ref{lem:intcommu}, we need two intermediary propositions, which are quite standard and whose proof is left to the reader.

\begin{proposition}\label{prop:intcommu1}
Let $G$ be a compact group, and $dg$ its normalized Haar measure. For any irreducible representation $(\rho,V)$ of $G$, we have
\begin{align*}
\int_G \chi_\rho(xgyg^{-1})dg=\frac{1}{d_\rho}\chi_\rho(x)\chi_\rho(y),\ \forall (x,y)\in G^2.
\end{align*}
\end{proposition}

\begin{proposition}\label{prop:intcommu2}
Let $G$ be a compact group and $(\rho,V)$, $(\pi,W)$ two irreducible representations of $G$. Then
\begin{align*}
\chi_\rho * \chi_\pi = \left\lbrace \begin{array}{cc}
\frac{\chi_\rho}{d_\rho} & \text{if } \rho\sim\pi,\\
0 & \text{otherwise.}
\end{array}\right.
\end{align*}
\end{proposition}

\begin{proof}[Proof of Lemma \ref{lem:intcommu}]
First, according to Proposition \ref{prop:intcommu1}, we have for any $(x,y)\in G^2$~:
\[ \int_G \chi_\rho(xyzy^{-1}z^{-1})dz = \frac{1}{d_\rho}\chi_\rho(xy)\chi_\rho(y^{-1}).\]
If we integrate out $y\in G$ it appears that
\[\int_{G^2} \chi_\rho(x[y,z])dydz = \frac{(\chi_\rho * \chi_\rho)(x)}{d_\rho},\]
which yields \eqref{eq02} using Proposition \ref{prop:intcommu2}.
\end{proof}

Finally, let us recall a well-known formula, which can be found for instance in \cite{Sta}.

\begin{lemma}[Pieri's rule]\label{lem:pieri}
Let $\lambda\in\widehat{\U}(N)$ be a highest weight, $r$ a positive integer. We have
\begin{equation}\label{eq:pieri2}
\Tr(x)s_\lambda(x) = \sum_{\substack{\mu\in\widehat{\U}(N)\\ \mu\searrow \lambda}} s_\mu(x),\ \forall x\in\U(N).
\end{equation}
\end{lemma}

We now have all the tools to prove Theorem \ref{prop:wilson_loops_exp_var}.

\begin{proof}[Proof of Theorem \ref{prop:wilson_loops_exp_var}]
Let us start from Eq. \eqref{eq:loop1}. We can decompose both heat kernels following \eqref{eq:FourierUN2}:
\begin{align*}
\E[\tr(H_\ell)]=&\frac{1}{Z_N(g,T)}\sum_{\lambda,\mu\in\widehat{\U}(N)}d_\lambda d_\mu e^{-\frac{c_2(\lambda)t}{2}-\frac{c_2(\mu)(T-t)}{2}}\\
&\int_{\U(N)^{2g+1}}\tr(x)s_\lambda(x^{-1})s_\mu(x[y_1,z_1]\cdots[y_g,z_g])dx\prod_{i=1}^gdy_idz_i.
\end{align*}
We can then apply Lemma \ref{lem:intcommu} $g$ times, which turns commutators into inverses of dimensions:
\begin{align*}
\E[\tr(H_\ell)]=&\frac{1}{Z_N(g,T)}\sum_{\lambda,\mu\in\widehat{\U}(N)}d_\lambda (d_\mu)^{1-2g} e^{-\frac{c_2(\lambda)t}{2}-\frac{c_2(\mu)(T-t)}{2}}\int_{\U(N)}\tr(x)s_\lambda(x^{-1})s_\mu(x)dx.
\end{align*}
Then, using Pieri's rule and the fact that $\tr=\frac{1}{N}\Tr$ gives
\begin{align*}
\E[\tr(H_\ell)]=&\frac{1}{Z_N(g,T)}\sum_{\lambda,\mu\in\widehat{\U}(N)}\frac{d_\lambda (d_\mu)^{1-2g}}{N} e^{-\frac{c_2(\lambda)t}{2}-\frac{c_2(\mu)(T-t)}{2}}\sum_{\substack{\nu\in\widehat{\U}(N)\\ \nu\searrow \mu}}\int_{\U(N)}s_\lambda(x^{-1})s_\nu(x)dx.
\end{align*}
From the orthogonality relations of characters, we deduce
\[
\mathbb{E}[\tr(H_\ell)]=\frac{1}{NZ_N(g,T)}\sum_{\substack{\lambda,\mu\in\widehat{\U}(N)\\\mu\nearrow\lambda}}\frac{e^{-\frac{T}{2}c_2(\mu)}}{(d_\mu)^{2g-2}}\frac{d_\lambda}{d_\mu} e^{\frac{t}{2}(c_2(\mu)-c_2(\lambda))}.
\]
We recover \eqref{eq:wilson_loop_exp_sun} by exchanging the roles of $\lambda$ and $\mu$ and rewriting the equation as an expectation over $\lambda$ with respect to $\mathbb{P}_{\widehat{\U}(N),T}$.\\

In the same manner as in Eq. \eqref{eq:loop1}, we can compute $\mathbb{E}[\vert\tr(H_\ell)\vert^2]=\mathbb{E}[\tr(H_\ell)\tr(H_\ell^*)]$ as
\begin{align*}
\mathbb{E}[\vert\tr(H_\ell)\vert^2] =\frac{1}{Z_N(g,T)}\int_{\U(N)^{2g+1}} \tr(x)\tr(x^{-1})p_t(x^{-1})p_{T-t}(x[y_1,z_1]\cdots[y_g,z_g])dx\prod_{i=1}^gdy_idz_i,
\end{align*}

which can be rewritten, using the heat kernel decomposition, Lemma \ref{lem:intcommu} and Lemma \ref{lem:pieri}:
\begin{align*}
\mathbb{E}[\vert\tr(H_\ell)\vert^2] = & \frac{1}{N^2Z_N(g,T)}\sum_{\lambda,\mu\in\widehat{\U}(N)}d_\lambda (d_\mu)^{1-2g}e^{-\frac{c_2(\lambda)t}{2}-\frac{c_2(\mu)(T-t)}{2}}\\
&\times\sum_{\substack{\nu,\tau\in\widehat{\U}(N)\\ \nu\searrow \mu, \tau\searrow\lambda}}\int_{\U(N)}s_\tau(x^{-1})s_\nu(x)dx.
\end{align*}

Using the orthogonality of Schur functions as before, we obtain Eq. \eqref{eq:wilson_loop_var_sun} as expected.
\end{proof}

\section{Large $N$ asymptotics}\label{sec:LargeN}

In this section, we prove Theorems \ref{thm:exp} and \ref{thm:var}. We need to control the convergence of partition functions first. Let us recall the following.

\begin{theorem}[\cite{Lem}]\label{thm:main} Let $\Sigma_{g,T}$ be an orientable surface of genus $g\geq 1$ and area $T\geq 0$. Then:
\begin{enumerate}
\item If $g\geq 2$ and $T>0$:
\begin{equation}
\lim_{N\to \infty} Z_{N}(g,T)=\theta(q_T).
\end{equation}
\item If $g=1$ and $T>0$:
\begin{equation}
\lim_{N\to \infty} Z_{N}(1,T)=\frac{\theta(q_T)}{\phi(q_T)^2}.
\end{equation}
\end{enumerate} 
\end{theorem}

We shall need a stronger statement in order to quantify the fluctuations.

\begin{lemma}\label{lem:fluct_pf}
Let $T>0$ be a fixed real number. We have uniformly, for $N$ large enough,
\begin{equation}
Z_N(1,T) - \frac{\theta(q_T)}{\phi(q_T)^2} = O(N^{-2}).
\end{equation}
For any $g\geq 2$, we have uniformly
\begin{equation}
Z_N(g,T) - \theta(q_T) = O(N^{2-2g}).
\end{equation}
\end{lemma}

\begin{proof}
The estimate for $g=1$ is a direct consequence of \cite[Theorem 1.2]{LM3}. Let us prove the case for $g\geq 2$. From the expression of the partition function we see that only flat highest weights contribute to the limit:
\[
\sum_{\lambda\in\Lambda_N}q_T^{c_2(\lambda)}d_\lambda^{2-2g} = \sum_{n\in\Z}q_T^{n^2}=\theta(q_T).
\]
We get then
\[
Z_N(g,T) - \theta(q_T) = \sum_{\lambda\in\widehat{\U}(N)\setminus\Lambda_N} q_T^{c_2(\lambda)}d_\lambda^{2-2g}.
\]
Recall that $d_\lambda\geq N$ for any $\lambda\not\in\Lambda_N$, therefore
\[
0\leq Z_N(g,T)-\theta(q_T) \leq N^{2-2g}\sum_{\lambda\in\widehat{\U}(N)\setminus\Lambda_N} q_T^{c_2(\lambda)}\leq N^{2-2g}\theta(q_T).
\]
The result follows directly.
\end{proof}

The previous lemma yields an immediate consequence, that will be useful for the fluctuations of Wilson loops.

\begin{corollary}\label{cor:fluct_pf}
For any $T>0$ and any $g\geq 2$, we have uniformly, for $N$ large enough,
\begin{equation}
\frac{Z_N(g,T)}{Z_N(1,T)} = \phi(q_T)^2 + O(N^{-2}).
\end{equation}
\end{corollary}

The proofs of Theorems \ref{thm:exp} and \ref{thm:var} will be split into two parts: we will first deal with $g\geq 2$ using flat highest weights, then we will prove them for $g=1$ with almost flat highest weights. In fact, all cases could be simultaneously handled with almost flat highest weights, but the proofs we give for $g\geq 2$ yield better estimates of the fluctuations, and show that the case $g=1$ is quite special.

\subsection{Proofs for $g\geq 2$}

\begin{proof}[Proof of Theorem \ref{thm:exp} for $g\geq 2$]
Let us start from Equation \eqref{eq:wilson_loop_exp_sun}. The sum over $\mu\in\widehat{\U}(N)$ can be split into flat and non-flat highest weights. We will show that the sum over $\Lambda_N$ gives the limit, and that the sum over $\widehat{\U}(N)\setminus\Lambda_N$ is $O(N^{-2})$. For any $\lambda=(n,\ldots,n)\in\Lambda_N$, the only $\mu\in\widehat{\U}(N)$ such that $\mu\searrow\lambda$ is $\mu=(n+1,n,\ldots,n)$, which has dimension $N$ and Casimir number
\[
c_2((n+1,n,\ldots,n))=n^2+1+\frac{2n}{N}.
\]
Furthermore, it is straightforward that $(n,\ldots,n)$ has dimension $1$ and Casimir number $n^2$.
It yields
\begin{align*}
\mathbb{E}[\tr(H_\ell)] = & \frac{1}{Z_N(g,T)}\sum_{n\in\Z}e^{-\frac{T}{2}n^2-\frac{t}{2}\left(1+\frac{2n}{N}\right)}\\
&+ \frac{1}{NZ_N(g,T)}\sum_{\substack{\lambda\in\widehat{\U}(N)\setminus \Lambda_N}} \frac{q^{c_2(\lambda)}}{d_\lambda^{2g-2}}\sum_{\substack{\mu\in\widehat{\U}(N)\\\mu\searrow\lambda}}\frac{d_\mu}{d_\lambda} e^{\frac{t}{2}(c_2(\mu)-c_2(\lambda))}.
\end{align*}
We rewrite it as $A_N+B_N$, where
\begin{equation}\label{eq:A_N}
A_N = \frac{1}{Z_N(g,T)}e^{-\frac{t}{2}} \sum_{n\in\Z} e^{-\frac{T}{2}n^2-\frac{t}{N}n}= \frac{1}{Z_N(g,T)}e^{-\frac{t}{2}+\frac{t^2}{2TN^2}} \sum_{n\in\Z} e^{-\frac{T}{2}\left(n+\frac{t}{TN}\right)^2},
\end{equation}
and
\[
B_N =\frac{1}{NZ_N(g,T)}\sum_{\substack{\lambda\in\widehat{\U}(N)\setminus \Lambda_N}} \frac{q^{c_2(\lambda)}}{d_\lambda^{2g-2}}\sum_{\substack{\mu\in\widehat{\U}(N)\\\mu\searrow\lambda}}\frac{d_\mu}{d_\lambda} e^{\frac{t}{2}(c_2(\mu)-c_2(\lambda))}.
\]
Using Lemma \ref{lem:fluct_pf}, we have
\[
A_N - q_t = \frac{1+O(N^{-2})}{\theta(q_T)}\sum_{n\in\Z}q_T^{n^2}q_t^{1+\frac{2n}{N}} - q_t.
\]
Putting $q_t/\theta(q_T)$ in factor and leaving the $O(N^{-2})$ term gives
\[
A_N - q_t = \frac{q_t}{\theta(q_T)}\left(\sum_{n\in\Z}q_T^{n^2}q_t^{\frac{2n}{N}} - \theta(q_T)\right) + O(N^{-2}).
\]
By Proposition \ref{prop:fluct_u1}, we finally obtain
\begin{equation}\label{eq:fluct1}
A_N - q_t = O(N^{-2}).
\end{equation}

Let us now deal with $B_N$: from Lemma 2.5 in \cite{Lem} and the fact that adding $1$ to all parts of a highest weight does not change the dimension, we get that
\[
d_\lambda\geq N,\ \forall \lambda\in\widehat{\U}(N)\setminus\Lambda_N.
\]
Furthermore, it is clear that $c_2(\lambda)\geq 0$ for any $\lambda$, from the definition of Casimir element. Thus,
\begin{align*}
0\leq B_N \leq & \frac{1}{NZ_N(g,T)}\sum_{\lambda\in\widehat{\U}(N)\setminus\Lambda_N} \frac{1}{N^{2g-2}}\sum_{\substack{\mu\in\widehat{\U}(N)\\\mu\searrow\lambda}} \frac{d_\mu}{d_\lambda} e^{-\frac{T}{2}c_2(\lambda)+\frac{t}{2}(c_2(\lambda)-c_2(\mu))}.
\end{align*}

Proposition \ref{prop:casimir2} implies that, for $N$ large enough,
\[
e^{-\frac{T}{2}c_2(\lambda)+\frac{t}{2}(c_2(\lambda)-c_2(\mu))}\leq e^{-\frac{T}{2}c_2(\lambda)+\frac{t}{N}(\vert\alpha\vert+\vert\beta\vert+2n)},
\]
where $\alpha,\beta,n$ are such that $\lambda=\lambda_N(\alpha,\beta,n)$. From \eqref{eq:pieri2} we have for any $\lambda\in\widehat{\U}(N)$
\[
\sum_{\mu\searrow\lambda} \frac{d_\mu}{d_\lambda}= N.
\]
These equations yield
\[
0\leq B_N\leq \frac{e^\frac{t}{2}}{Z_N(g,T)}\frac{1}{N^{2g-2}}\sum_{\lambda\in\widehat{\U}(N)\setminus\Lambda_N}e^{-\frac{T}{2}c_2(\lambda)+\frac{t}{N}(\vert\alpha\vert+\vert\beta\vert+2n)}.
\]
According to Lemma \ref{lem:domin}, the sum on the right is bounded by
\[
\sum_{\substack{\alpha,\beta\in\mathfrak{Y}:\\\ell(\alpha)\leq N/2\\\ell(\beta)\leq N/2}}\sum_{n\in\Z}e^{-\frac{T}{2}C_N(\alpha,\beta,n)+\frac{t}{N}(\vert\alpha\vert+\vert\beta\vert+2n)} = \sum_{\substack{\alpha,\beta\in\mathfrak{Y}:\\\ell(\alpha)\leq N/2\\\ell(\beta)\leq N/2}} e^{-(\vert\alpha\vert+\vert\beta\vert)\left(\frac{T}{2}-\frac{t}{N}\right)}\sum_{n\in\Z}e^{-\frac{T}{2}\left(n+\frac{\vert\alpha\vert-\vert\beta\vert}{N}\right)^2+\frac{2tn}{N}}.
\]
The sum over $n$ is bounded independently from $\alpha,\beta$ by
\[
C_1=\sup_{x\in[0,1]}\sum_{n\in\Z}e^{-\frac{T}{2}(x+n)^2+\frac{2tn}{N}}\leq C_2=\sup_{x\in[0,1]}\sum_{n\in\Z}q_T^{(x+n)^2}.
\]
The last bound is independent from $N$, and if we notice that $Z_N(g,T)\geq 1$ for all $N\geq 1$, we finally get
\[
0\leq B_N\leq N^{2-2g}e^{\frac{t}{2}}\left(\sum_{\alpha}e^{-\vert\alpha\vert\left(\frac{T}{2}-\frac{t}{N}\right)}\right)C_2.
\]
For $N$ large enough, we have $\frac{T}{2}-\frac{t}{N}\geq 1/2$ and the sum on $\alpha$ becomes uniformly bounded for larger values of $N$, and thus $B_N=O(N^{2-2g})=O(N^{-2})$. If we combine this estimate with \eqref{eq:fluct1}, we obtain both \eqref{eq:limexp} and \eqref{eq:fluct_exp}.
\end{proof}

\begin{proof}[Proof of Theorem \ref{thm:var} for $g\geq 2$]
Starting from Equation \eqref{eq:wilson_loop_var_sun}, we have
\[
\mathbb{E}[\vert\tr(H_\ell)\vert^2] = A_N+B_N,
\]
where
\[
A_N = \frac{Z_N(1,T)}{Z_N(g,T)}\E_{\widehat{\U}(N),T}\left[\mathbf{1}_{\Lambda_N}(\lambda)\sum_{\substack{\mu\in\widehat{\U}(N): \mu\sim\lambda}}\frac{d_\mu}{N^2d_\lambda^{2g-1}}q_t^{c_2(\mu)-c_2(\lambda)} \right]
\]
and
\[
B_N = \frac{Z_N(1,T)}{Z_N(g,T)}\E_{\widehat{\U}(N),T}\left[\mathbf{1}_{\widehat{\U}(N)\setminus\Lambda_N}(\lambda)\sum_{\substack{\mu\in\widehat{\U}(N): \mu\sim\lambda}}\frac{d_\mu}{N^2d_\lambda^{2g-1}}q_t^{c_2(\mu)-c_2(\lambda)} \right].
\]
If $\lambda=(n,\ldots,n)\in\Lambda_N$, then there are only two highest weights equivalent to $\lambda$, which are $\lambda'=(n+1,n,\ldots,n,n-1)$ and $\lambda$ itself (as soon as $N\geq 2$). After some quick computations, we get $c_2(\lambda')=n^2+2$ and $d_{\lambda'}=N^2-1$, therefore
\[
\E_{\widehat{\U}(N),T}\left[\mathbf{1}_{\Lambda_N}(\lambda)\sum_{\substack{\mu\in\widehat{\U}(N): \mu\sim\lambda}}\frac{d_\mu}{N^2d_\lambda^{2g-1}}q_t^{c_2(\mu)-c_2(\lambda)} \right]=\frac{1}{Z_N(1,T)}\sum_{n\in\Z}q_T^{n^2}\left(\frac{1}{N^2}+e^{-t}\frac{N^2-1}{N^2}\right).
\]
If we combine this equality with Lemma \ref{lem:fluct_pf}, we finally obtain
\[
A_N = \frac{\theta(q_T)}{Z_N(g,T)}\left(\frac{1}{N^2}+e^{-t}\frac{N^2-1}{N^2}\right)=\left(1+O(N^{-2})\right)\left(N^{-2}+e^{-t}(1+N^{-2})\right)=e^{-t}\left(1+O(N^{-2})\right).
\]

Using the same arguments as in the proof of Theorem \ref{thm:exp}, we can prove that $B_N=O(N^{-2})$. Combining the estimates of $A_N$ and $B_N$ leads to
\[
\lim_{N\to\infty}\left(\E[\vert\tr(H_\ell)\vert^2]-\vert\E[\tr(H_\ell)]\vert^2\right) = O(N^{-2}),
\]
which concludes the proof.
\end{proof}

\subsection{Proofs for $g=1$}

Before diving into the proof of Theorems \ref{thm:exp} and \ref{thm:var}, we must derive a few intermediate results, in order to express the moments of Wilson loops in terms of almost flat highest weights. We will first prove in Lemma \ref{lem:bound_lambdagamma} that the probability that $\lambda_N(\alpha,\beta,n)$ is not almost flat is bounded by an arbitrary power of $\frac1N$ for $N$ large enough, then we will express in Lemma \ref{cor:Wilson_loop_exp_partition} the Wilson loop expectation in terms of $(\alpha,\beta,n)\in\Lambda_N^\gamma$ for $\gamma>0$ small enough. Restricting ourselves to $\Lambda_N^\gamma$ will be mandatory to derive the right bounds for the differences of Casimir numbers and the ratios of dimensions.

\begin{lemma}\label{lem:bound_lambdagamma}
For any $T>0$, $\gamma\in(0,\frac13)$ and any $k\geq 1$, if $(\alpha,\beta,n)$ are independent random variables such that $(\alpha,\beta)\sim\mathscr{U}(q_T)$ and $n\sim\mathbb{P}_{\widehat{\U}(1),T}$, then as $N\to\infty$,
\begin{equation}
\mathbb{P}((\alpha,\beta,n)\not\in\Lambda_N^\gamma) = O(N^{-k}).
\end{equation}
\end{lemma}

\begin{proof}
We note that
\[
\mathfrak{Y}\times\mathfrak{Y}\times\Z\setminus\Lambda_N^\gamma \subset \{(\alpha,\beta,n): \vert\alpha\vert > N^\gamma\}\cup\{(\alpha,\beta,n):\vert\beta\vert>N^\gamma\}.
\]
Taking probabilities, we find
\[
\mathbb{P}((\alpha,\beta,n)\not\in\Lambda_N^\gamma)\leq \mathbb{P}(\vert\alpha\vert>N^\gamma)+\mathbb{P}(\vert\beta\vert>N^\gamma) = 2\mathbb{P}(\vert\alpha\vert>N^\gamma).
\]
From Proposition \ref{prop:dev_ineq}, there exists $C_T>0$ such that
\[
0\leq \mathbb{P}((\alpha,\beta,n)\not\in\Lambda_N^\gamma)\leq 2C_Te^{-\frac{T}{2}N^\gamma}=O(N^{-k}),
\]
as expected.
\end{proof}

\begin{lemma}\label{cor:Wilson_loop_exp_partition}
Let $(\alpha,\beta,n)$ be independent random variables such that $\alpha,\beta\sim\mathscr{U}(q_T)$ and $n\sim\mathbb{P}_{\widehat{\U}(1),T}$. Let $\Sigma_{1,T}$ be an orientable compact connected surface of genus $1$ and of area $T>0$, $\ell$ be a contractible simple loop of interior area $t$ oriented clockwise. The following holds true for any fixed $\gamma\in(0,\frac13)$ and for $N$ large enough:
\begin{equation}
\E[\tr(H_\ell)] = \frac{\theta(q_T)}{Z_N(1,T)\phi(q_T)^2}\E\left[\mathbf{1}_{\Lambda_N^\gamma}(\alpha,\beta,n)\Psi_N(\alpha,\beta,n)\right](1+O(N^{3\gamma-1})),
\end{equation}
where
\[
\Psi_N(\alpha,\beta,n)=q_t^{\frac{2}{N}F(\alpha,\beta,n)}\left[\sum_{\alpha'\searrow\alpha}\frac{d^{\alpha'}}{\vert\alpha'\vert d^\alpha}q_t^{1+\frac{2}{N}(c(\square')+n)}+\frac{\vert\beta\vert}{N^2}\sum_{\beta'\nearrow\beta}\frac{d^{\beta'}}{d^\beta}q_t^{-1+\frac{2}{N}(n-c(\square'))}\right].
\]
\end{lemma}

\begin{proof}
From Theorem \ref{prop:wilson_loops_exp_var}, we have
\[
\mathbb{E}[\tr(H_\ell)]= \E_{\widehat{\U}(N),T}\left[\sum_{\substack{\mu\in\widehat{\U}(N): \mu\searrow\lambda}}\frac{d_\mu}{Nd_\lambda}q_t^{c_2(\mu)-c_2(\lambda)} \right].
\]
Let us split $\widehat{\U}(N)$ into
\[
\Omega_N^{1}=\{\lambda\in\widehat{\U}(N):(\alpha_\lambda,\beta_\lambda,n_\lambda)\in\Lambda_N^\gamma\}\ \text{and} \ \Omega_N^{2}=\{\lambda\in\widehat{\U}(N):(\alpha_\lambda,\beta_\lambda,n_\lambda)\not\in\Lambda_N^\gamma\},
\]
meaning that
\[
\E[\tr(H_\ell)]=E_1+E_2,
\]
with
\[
E_i = \E_{\widehat{\U}(N),T}\left[\mathbf{1}_{\Omega_N^i}(\lambda)\sum_{\substack{\mu\in\widehat{\U}(N): \mu\searrow\lambda}}\frac{d_\mu}{Nd_\lambda}q_t^{c_2(\mu)-c_2(\lambda)}\right],\quad i\in\{1,2\}.
\]
Using Propostion \ref{prop:change_variable_partition}, Proposition \ref{prop:casimir2} and Lemma \ref{lem:ratio_dim}, we get
\[
E_1=\frac{\theta(q_T)}{Z_N(1,T)\phi(q_T)^2}\E[\mathbf{1}_{\Lambda_N^\gamma}(\alpha,\beta,n)\Psi_N(\alpha,\beta,n)](1+O(N^{3\gamma-1})).
\]
It remains to prove that the $E_2$ is negligible. First, from Proposition \ref{prop:casimir2} we get that
\[
c_2(\mu)-c_2(\lambda)\geq -1-\frac{1}{N}(\vert\alpha_\lambda\vert+\vert\beta_\lambda\vert+2n_\lambda),
\]
which only depends on $\lambda$, therefore
\[
E_2\leq \E_{\widehat{\U}(N),T}\left[\mathbf{1}_{\Omega_N^2}(\lambda)q_t^{-1-\frac{1}{N}(\vert\alpha_\lambda\vert+\vert\beta_\lambda\vert+2n)}\frac{1}{N}\sum_{\substack{\mu\in\widehat{\U}(N): \mu\searrow\lambda}}\frac{d_\mu}{d_\lambda}\right].
\]
For any $\lambda$, we deduce from Lemma \ref{lem:pieri} that
\[
\sum_{\mu\searrow\lambda}\frac{d_\mu}{d_\lambda}=N,
\]
so that
\[
E_2\leq\E_{\widehat{\U}(N),T}\left[\mathbf{1}_{\Omega_N^2}(\lambda)q_t^{-1-\frac{1}{N}(\vert\alpha_\lambda\vert+\vert\beta_\lambda\vert+2n)}\right].
\]
Using Proposition \ref{prop:change_variable_partition}, we get
\[
E_2 \leq \frac{\theta(q_T)}{Z_N(1,T)\phi(q_T)^2}\E\left[\mathbf{1}_{\Lambda_N^0\setminus\Lambda_N^\gamma}(\alpha,\beta,n)q_T^{\frac{2}{N}F(\alpha,\beta,n)}q_t^{-1-\frac{1}{N}(\vert\alpha\vert+\vert\beta\vert+2n)}\right].
\]
The expectation in the RHS can be bounded by the Cauchy--Schwarz inequality
\begin{align*}
\E\big[\mathbf{1}_{(\Lambda_N^\gamma)^c}(\alpha,\beta,n)&\mathbf{1}_{\Lambda_N^0}(\alpha,\beta,n)q_T^{\frac{2}{N}F(\alpha,\beta,n)}q_t^{-1-\frac{1}{N}(\vert\alpha\vert+\vert\beta\vert+2n)}\big]\\
\leq & \sqrt{\mathbb{P}((\alpha,\beta,n)\not\in\Lambda_N^\gamma)\E\big[\mathbf{1}_{\Lambda_N^0}(\alpha,\beta,n)q_{2T}^{\frac{2}{N}F(\alpha,\beta,n)}q_{2t}^{-1-\frac{1}{N}(\vert\alpha\vert+\vert\beta\vert+2n)}\big]}.
\end{align*}
We can apply Lemma \ref{lem:bound_lambdagamma} for $k=4$ to get
\[
E_2\leq O(N^{-2})\frac{\theta(q_T)}{Z_N(1,T)\phi(q_T)^2}\sqrt{\E\big[\mathbf{1}_{\Lambda_N^0}(\alpha,\beta,n)q_{2T}^{\frac{2}{N}F(\alpha,\beta,n)}q_{2t}^{-1-\frac{1}{N}(\vert\alpha\vert+\vert\beta\vert+2n)}\big]}.
\]
However, the quantity
\[
\frac{\theta(q_T)}{Z_N(1,T)\phi(q_T)^2}\sqrt{\E\big[\mathbf{1}_{\Lambda_N^0}(\alpha,\beta,n)q_{2T}^{\frac{2}{N}F(\alpha,\beta,n)}q_{2t}^{-1-\frac{1}{N}(\vert\alpha\vert+\vert\beta\vert+2n)}\big]}
\]
converges to a finite value as $N\to\infty$, therefore it is bounded uniformly in $N$. We obtain that $E_2=O(N^{-2})$.
\end{proof}

We can now turn to the proof of Theorem \ref{thm:exp} for $g=1$.

\begin{proof}[Proof of Theorem \ref{thm:exp} for $g=1$]
Let $\gamma\in(0,\frac13)$ be a fixed real number. From Lemma \ref{cor:Wilson_loop_exp_partition} and Corollary \ref{cor:fluct_pf},
\[
\vert\E[\tr(H_\ell)]-q_t\vert = \left\vert\E[\mathbf{1}_{\Lambda_N^\gamma}(\alpha,\beta,n)\Psi_N(\alpha,\beta,n)]-q_t\right\vert(1+O(N^{-2})(1+O(N^{3\gamma-1})).
\]
The proof therefore boils down to prove
\begin{equation}\label{eq:estimate_thm12_g1}
\left\vert\E[\mathbf{1}_{\Lambda_N^\gamma}(\alpha,\beta,n)\Psi_N(\alpha,\beta,n)]-q_t\right\vert = O(N^{\gamma-1}).
\end{equation}
In fact,
\begin{align*}
\vert\E[\mathbf{1}_{\Lambda_N^\gamma}(\alpha,\beta,n)\Psi_N(\alpha,\beta,n)]-q_t\vert= &\vert\E[\mathbf{1}_{\Lambda_N^\gamma}(\alpha,\beta,n)\left(\Psi_N(\alpha,\beta,n)-q_t\right)]-q_t\E[\mathbf{1}_{\Lambda_N^0\setminus\Lambda_N^\gamma}(\alpha,\beta,n)]\vert\\
\leq & \vert\E[\mathbf{1}_{\Lambda_N^\gamma}(\alpha,\beta,n)\left(\Psi_N(\alpha,\beta,n)-q_t\right)]\vert + q_t\mathbb{P}((\alpha,\beta,n)\in\Lambda_N^0\setminus\Lambda_N^\gamma).
\end{align*}
First, from Lemma \ref{lem:bound_lambdagamma} we have
\[
q_t\mathbb{P}((\alpha,\beta,n)\in\Lambda_N^0\setminus\Lambda_N^\gamma) = O(N^{-2}).
\]
It remains to bound the other term. For any $(\alpha,\beta,n)\in\Lambda_N^\gamma$, we have
\[
\Psi_N(\alpha,\beta,n)=\sum_{\alpha'\searrow\alpha}\frac{d^{\alpha'}}{\vert\alpha'\vert d^\alpha}q_t^{1+\frac{2}{N}(c(\square')+n)}+\frac{\vert\beta\vert}{N^2}\sum_{\beta'\nearrow\beta}\frac{d^{\beta'}}{d^\beta}q_t^{-1+\frac{2}{N}(n-c(\square'))}.
\]
For any $\alpha'\searrow\alpha$, if we denote by $\square'$ the box added to $\alpha$ to obtain $\alpha'$, we find that
\[
\vert c(\square')\vert \leq \max\{\alpha_1+1,\ell(\alpha)+1\}\leq N^\gamma.
\]
In particular,
\[
\left\vert\frac{c(\square')}{N}\right\vert \leq N^{\gamma-1}.
\]
Hence,
\[
\sum_{\alpha'\searrow\alpha}\frac{d^{\alpha'}}{\vert\alpha'\vert d^\alpha}q_t^{1+\frac{2}{N}(c(\square')+n)}\leq \sum_{\alpha'\searrow\alpha}\frac{d^{\alpha'}}{\vert\alpha'\vert d^\alpha}q_t^{1+\frac{2n}{N}-2N^{\gamma-1}} = q_t^{1+\frac{2n}{N}-2N^{\gamma-1}}.
\]
Analogously,
\[
\sum_{\alpha'\searrow\alpha}\frac{d^{\alpha'}}{\vert\alpha'\vert d^\alpha}q_t^{1+\frac{2}{N}(c(\square')+n)}\geq q_t^{1+\frac{2n}{N}+2N^{\gamma-1}}.
\]
From these inequalities, we find that
\[
\sum_{\alpha'\searrow\alpha}\frac{d^{\alpha'}}{\vert\alpha'\vert d^\alpha}q_t^{1+\frac{2}{N}(c(\square')+n)} = q_t^{1+\frac{2n}{N}}(1+O(N^{\gamma-1})),
\]
where the $O$ is uniform in $\alpha,\beta,n$. We also get, with the same arguments,
\[
\frac{\vert\beta\vert}{N^2}\sum_{\beta'\nearrow\beta}\frac{d^{\beta'}}{d^\beta}q_t^{-1+\frac{2n}{N}-\frac{2c(\square')}{N}} = O(N^{\gamma-2})q_t^{-1+\frac{2n}{N}}(1+O(N^{\gamma-1})) = q_t^{-1+\frac{2n}{N}}O(N^{\gamma-2}).
\]
Combining these estimates leads to
\begin{equation}\label{eq:estimate_FN}
\Psi_N(\alpha,\beta,n)= q_t^{1+\frac{2n}{N}}(1+O(N^{\gamma-1}))+q_t^{-1+\frac{2n}{N}}O(N^{\gamma-2}).
\end{equation}
From there,
\[
\E[\mathbf{1}_{\Lambda_N^\gamma}(\alpha,\beta,n)(\Psi_N(\alpha,\beta,n)-q_t)] = q_t\E\left[\mathbf{1}_{\Lambda_N^\gamma}(\alpha,\beta,n)\left(q_t^{\frac{2n}{N}}(1+O(N^{\gamma-1}))+q_t^{-2+\frac{2n}{N}}O(N^{\gamma-2})-1\right)\right].
\]
\begin{align*}
\E[\mathbf{1}_{\Lambda_N^\gamma}(\alpha,\beta,n)(\Psi_N(\alpha,\beta,n)-q_t)] = & q_t\E\left[\mathbf{1}_{\Lambda_N^\gamma}(\alpha,\beta,n)\left(q_t^{\frac{2n}{N}}-1\right)\right]+ q_t\E[\mathbf{1}_{\Lambda_N^\gamma}(\alpha,\beta,n)q_t^{\frac{2n}{N}}]O(N^{\gamma-1})\\
& + \E[\mathbf{1}_{\Lambda_N^\gamma}(\alpha,\beta,n)q_t^{-1+\frac{2n}{N}}]O(N^{\gamma-2}).
\end{align*}
Using Proposition \ref{prop:fluct_u1}, we find that
\[
\E[\mathbf{1}_{\Lambda_N^\gamma}(\alpha,\beta,n)(\Psi_N(\alpha,\beta,n)-q_t)] = O(N^{\gamma-1}),
\]
which yields \eqref{eq:estimate_thm12_g1}.
\end{proof}

\begin{proof}[Proof of Theorem \ref{thm:var} for $g=1$]
Given $(\alpha,\beta,n)\in\Lambda_N^\gamma$, set
\[
\Lambda_N^{(1)}(\alpha,\beta,n) = \{(\alpha',\beta',n), \alpha'\searrow\alpha,\beta'\searrow\beta\},
\]
\[
\Lambda_N^{(2)}(\alpha,\beta,n) = \{(\alpha',\beta',n), \alpha'\nearrow\alpha,\beta'\nearrow\beta\},
\]
\[
\Lambda_N^{(3)}(\alpha,\beta,n) = \{(\alpha,\beta',n), \beta'\sim\beta,\beta'\neq\beta\},
\]
\[
\Lambda_N^{(4)}(\alpha,\beta,n) = \{(\alpha',\beta,n), \alpha'\sim\alpha,\alpha'\neq\alpha\},
\]
\[
\Lambda_{N}^{(5)}(\alpha,\beta,n)=\{(\alpha,\beta,n)\}.
\]
Using the same arguments as in the proof of Lemma \ref{cor:Wilson_loop_exp_partition}, we can obtain
\[
\E[\vert\tr(H_\ell)\vert^2] = (E_1+E_2+E_3+E_4+E_5),
\]
with
\[
E_i = \frac{\theta(q_T)}{N^2Z_N(1,T)\phi(q_T)^2}\E\left[\mathbf{1}_{\Lambda_N^\gamma}(\alpha,\beta,n)\sum_{(\alpha',\beta',n)\in\Lambda_N^{(i)}}\frac{d_{\lambda_N(\alpha',\beta',n)}}{d_{\lambda_N(\alpha,\beta,n)}}q_t^{c_2(\lambda_N(\alpha',\beta',n))-c_2(\lambda_N(\alpha,\beta,n))}\right].
\]
From Lemma \ref{lem:fluct_pf} we get
\[
E_i = \frac{1+O(N^{-2})}{N^2}\E\left[\mathbf{1}_{\Lambda_N^\gamma}(\alpha,\beta,n)\sum_{(\alpha',\beta',n)\in\Lambda_N^{(i)}}\frac{d_{\lambda_N(\alpha',\beta',n)}}{d_{\lambda_N(\alpha,\beta,n)}}q_t^{c_2(\lambda_N(\alpha',\beta',n))-c_2(\lambda_N(\alpha,\beta,n))}\right].
\]
We shall first estimate $E_1$, then prove that $E_i=O(N^{-2})$ for $i\geq 2$.
\[
E_1=\frac{1+O(N^{-2})}{N^2}\E\left[\mathbf{1}_{\Lambda_N^\gamma}(\alpha,\beta,n)\sum_{\alpha'\searrow\alpha}\sum_{\beta'\searrow\beta}\frac{d_{\lambda_N(\alpha',\beta',n)}}{d_{\lambda_N(\alpha,\beta,n)}}q_t^{c_2(\lambda_N(\alpha',\beta',n))-c_2(\lambda_N(\alpha,\beta,n))}\right].
\]
Using Proposition \ref{prop:casimir2} and Lemma \ref{lem:ratio_dim}, we find on the one hand
\[
E_1 \leq \frac{e^{-t-\frac{t}{2}2N^{\gamma-1}}}{N^2}\E\left[\mathbf{1}_{\Lambda_N^\gamma}(\alpha,\beta,n)\sum_{\alpha'\searrow\alpha}\sum_{\beta'\searrow\beta}
\frac{N^2 d^{\beta'}d^{\alpha'}}{\vert\beta\vert d^{\beta}\vert\alpha\vert d^{\alpha}}\right](1+O(N^{3\gamma-1})).
\]
It can be simplified thanks to Proposition \ref{prop:branch} and Lemma \ref{lem:bound_lambdagamma}:
\[
E_1 \leq e^{-t+\frac{t}{2}2N^{\gamma-1}}\mathbb{P}((\alpha,\beta,n)\in\Lambda_N^\gamma)(1+O(N^{3\gamma-1}))=e^{-t}(1+O(N^{3\gamma-1})).
\]
On the other hand, we also have
\[
E_1\geq e^{-t-\frac{t}{2}2N^{\gamma-1}}\mathbb{P}((\alpha,\beta,n)\in\Lambda_N^\gamma)(1+O(N^{3\gamma-1}))=e^{-t}(1+O(N^{3\gamma-1})).
\]
It follows that $E_1=e^{-t}(1+O(N^{3\gamma-1}))$ as $N\to\infty$. Now, it remains to show that $E_i=O(N^{-2})$ for $i\geq 2$.
\[
E_2 = \frac{1+O(N^{-2})}{N^2}\E\left[\mathbf{1}_{\Lambda_N^\gamma}(\alpha,\beta,n)\sum_{\alpha'\nearrow\alpha}\sum_{\beta'\nearrow\beta}\frac{d_{\lambda_N(\alpha',\beta',n)}}{d_{\lambda_N(\alpha,\beta,n)}}q_t^{c_2(\lambda_N(\alpha',\beta',n))-c_2(\lambda_N(\alpha,\beta,n))}\right].
\]
Using Proposition \ref{prop:casimir2}, Lemma \ref{lem:ratio_dim} and Proposition \ref{prop:branch}, we find
\begin{align*}
E_2\leq & \frac{e^{-t-\frac{t}{2}2N^{\gamma-1}}}{N^2}\E\left[\mathbf{1}_{\Lambda_N^\gamma}(\alpha,\beta,n)\sum_{\alpha'\nearrow\alpha}\sum_{\beta'\nearrow\beta}
\frac{\vert\beta\vert \vert\alpha\vert d^{\beta'}d^{\alpha'}}{N^2 d^{\beta}d^{\alpha}}\right](1+O(N^{3\gamma-1}))\\
\leq & O(N^{-4})\E\left[\mathbf{1}_{\Lambda_N^\gamma}(\alpha,\beta,n)
 \vert\alpha\vert\vert\beta\vert\right] = O(N^{-4}).
\end{align*}
We also get
\[
E_3 = \frac{1+O(N^{-2})}{N^2} \E\left[\mathbf{1}_{\Lambda_N^\gamma}(\alpha,\beta,n)\sum_{\substack{\alpha'\sim\alpha\\\alpha'\neq\alpha}}\frac{d_{\lambda_N(\alpha',\beta,n)}}{d_{\lambda_N(\alpha,\beta,n)}}q_t^{c_2(\lambda_N(\alpha',\beta,n))-c_2(\lambda_N(\alpha,\beta,n))}\right].
\]
However, using Proposition \ref{prop:casimir2}, Lemma \ref{lem:ratio_dim} and Proposition \ref{prop:branch},
\begin{align*}
1 + \sum_{\substack{\alpha'\sim\alpha\\\alpha'\neq\alpha}}&\frac{d_{\lambda_N(\alpha',\beta,n)}q_t^{c_2(\lambda_N(\alpha',\beta,n))}}{d_{\lambda_N(\alpha,\beta,n)}q_t^{c_2(\lambda_N(\alpha,\beta,n))}} \\
= & \sum_{\alpha''\searrow\alpha}\sum_{\alpha'\nearrow\alpha''}\frac{d_{\lambda_N(\alpha'',\beta,n)}q_t^{c_2(\lambda_N(\alpha'',\beta,n))}}{d_{\lambda_N(\alpha,\beta,n)}q_t^{c_2(\lambda_N(\alpha,\beta,n))}}\frac{d_{\lambda_N(\alpha',\beta,n)}q_t^{c_2(\lambda_N(\alpha',\beta,n))}}{d_{\lambda_N(\alpha'',\beta,n)}q_t^{c_2(\lambda_N(\alpha'',\beta,n))}}\\
\leq & q_t^{2-4N^{\gamma-1}}\sum_{\alpha''\searrow\alpha}\sum_{\alpha'\nearrow\alpha''}\frac{Nd^{\alpha''}}{\vert\alpha''\vert d^{\alpha}}\frac{\vert\alpha''\vert d^{\alpha'}}{Nd^{\alpha''}}(1+O(N^{3\gamma-1}))\\
= & q_t^{2-4N^{\gamma-1}}(\vert\alpha\vert+1)(1+O(N^{3\gamma-1})).
\end{align*}
Hence, putting it back into the expectation yields
\[
E_3 = O(N^{-2}).
\]
The estimation of $E_4$ is exactly the same as $E_3$, and finally we have
\[
E_5 = \frac{1+O(N^{-2})}{N^2}\E[\mathbf{1}_{\Lambda_N^\gamma}(\alpha,\beta,n)] = O(N^{-2}),
\]
which concludes the proof. Indeed, combining all estimates from $E_i$, we have
\begin{equation}
\E[\vert\tr(H_\ell)\vert^2] = e^{-t}(1+O(N^{3\gamma-1}))+O(N^{-2})=e^{-t}+O(N^{3\gamma-1}).
\end{equation}
\end{proof}

\section{Large $g$ asymptotics}\label{sec:LargeG}

In this section, we will prove Theorem \ref{thm:large_g}. The main idea will be the same as for Theorem \ref{thm:exp} and \ref{thm:var} for $g\geq 2$, which is that only flat highest weights contribute to the limit. As in Section \ref{sec:LargeN}, we first need an estimate of the partition function. 

\begin{theorem}
For any $T>0$ and any fixed integers $N\geq 2$ and $k\geq 1$, as $g\to\infty$,
\begin{equation}
Z_N(g,T)= \theta(q_T) + O(g^{-k}).
\end{equation}
\end{theorem}

\begin{proof}
Recall that
\[
Z_N(g,T) = \sum_{\lambda\in\widehat{\U}(N)} q_t^{c_2(\lambda)}d_\lambda^{2-2g} = \sum_{n\in\Z} q_T^{n^2} + \sum_{\lambda\in\widehat{\U}(N)\setminus\Lambda_N} q_T^{c_2(\lambda)}d_\lambda^{2-2g}.
\]
However, as $d_\lambda\geq N$ for all $\lambda\notin\Lambda_N$,
\[
0\leq \sum_{\lambda\in\widehat{\U}(N)\setminus\Lambda_N} q_T^{c_2(\lambda)}d_\lambda^{2-2g}\leq N^{2-2g}\sum_{\lambda\in\widehat{\U}(N)}q_T^{c_2(\lambda)}=N^{2-2g}Z_N(1,T).
\]
As a consequence,
\[
Z_N(g,T) = \theta(q_T) + O(N^{2-2g}),
\]
as $g\to\infty$. For any fixed $N$ and $k$, as $g\to\infty,$ $N^{2-2g}$ is dominated by $g^{-k}$, therefore the result follows.
\end{proof}

\begin{proof}[Proof of Theorem \ref{thm:large_g}]
Let us start with the formula \eqref{eq:wilson_loop_exp_sun}. We can split the expectation on $\Lambda_N$ and $\widehat{\U}(N)\setminus\Lambda_N$ as in previous proofs.
\[
\E[\tr(H_\ell)] = \frac{1}{Z_N(g,T)}(A + B_g),
\]
where
\[
A = \sum_{n\in\Z} q_T^{n^2} \sum_{\mu\in\widehat{\U}(N):\mu\searrow(n,\ldots,n)}\frac{d_\mu}{N} q_t^{c_2(\mu)-n^2},
\]
which does not depend on $g$, and
\[
B_g = \sum_{\lambda\in\widehat{\U}(N)\setminus\Lambda_N} q_T^{c_2(\lambda)} \sum_{\mu\searrow\lambda} \frac{d_\mu}{Nd_\lambda^{2g-1}}.
\]
Using again the bound $d_\lambda\geq N$ for $\lambda\not\in\Lambda_N$, we have
\[
B_g \leq N^{2-2g} \sum_{\lambda\in\widehat{\U}(N)\setminus\Lambda_N} q_T^{c_2(\lambda)} \sum_{\mu\searrow\lambda} \frac{d_\mu}{Nd_\lambda}.
\]
From there, it is quite obvious that for any $k$, $B_g=O(g^{-k})$ as $g\to\infty$. Now, let us compute $A$ more precisely: given $n\in\Z$, there is only one $\mu\searrow(n,\ldots,n)$, which is $(n+1,n,\ldots,n)$. It can be written as $\lambda_N(\alpha,\beta,n)$ with $\alpha=(1)$ and $\beta=\varnothing$. It has dimension $N$ and
\[
c_2(\mu)-c_2(n,\ldots,n) = 1 + \frac{2n}{N}.
\]
Hence,
\[
A = \sum_{n\in\Z} q_T^{n^2}  q_t^{1 + \frac{2n}{N}}.
\]
We then obtain
\[
\E[\tr(H_\ell)] = \frac{1}{\theta(q_T)}\sum_{n\in\Z} q_T^{n^2} q_t^{\frac{2n}{N}+1} + O(g^{-k}),
\]
which implies \eqref{eq:exp_largeg} if we use again Theorem \ref{thm:exp}.\\

The proof of \eqref{eq:var_largeg} is quite the same, but derived from \eqref{eq:wilson_loop_var_sun}: we obtain
\[
\E[\vert\tr(H_\ell)\vert^2] = \frac{1}{Z_N(g,T)}(A'+B'_g),
\]
where
\[
A'=\sum_{n\in\Z}q_T^{n^2}\sum_{\mu\sim(n,\ldots,n)}\frac{d_\mu}{N^2}q_t^{c_2(\mu)-n^2},
\]
and
\[
B'_g=O(g^{-k}),\quad \forall k\geq 1.
\]
In order to compute $A'$, let us notice that if $N\geq 2$, two highest weights of $\U(N)$ satisfy $\mu\sim(n,\ldots,n)$: there is $(n,\ldots,n)$ itself, as well as $(n+1,n,\ldots,n,n-1)$. The latter corresponds to $\lambda_N(\alpha,\beta,n)$ with $\alpha=\beta=(1)$. We have
\[
c_2(n+1,n,\ldots,n,n-1)-c_2(n,\ldots,n) = 2. 
\]
Furthermore, one finds that $d_{\lambda_N((1),(1),n)}=N^2$. The expression of $A'$ can be simplified as follows:
\[
A' = \sum_{n\in\Z}q_T^{n^2}\left(\frac{1}{N^2}+q_t^2\right) = \theta(q_T)\left(\frac{1}{N^2}+e^{-t}\right).
\]
The result follows.
\end{proof}

\bibliographystyle{alpha}
\bibliography{AFHW}

\end{document}